\documentclass[10pt]{article}
\usepackage{geometry}
\usepackage{graphicx}
\graphicspath{graphics/}
\usepackage{amsfonts}
\usepackage{mathrsfs}
\usepackage{float}
\usepackage{bm}
\usepackage{nicematrix}
\usepackage{mathdots}
\usepackage{mathrsfs}
\usepackage{array}
\usepackage{tikz}
\usepackage{stmaryrd}
\usepackage{verbatim}
\usepackage{tikz-cd}
\usepackage{booktabs}
\usepackage{enumitem}
\usepackage{amsmath}
\usepackage{amssymb}
\usepackage{amsthm}

\usepackage{hyperref}

\theoremstyle{plain}
\newtheorem{theorem}{Theorem}[section]
\newtheorem{proposition}[theorem]{Proposition}
\newtheorem{lemma}[theorem]{Lemma}
\newtheorem{corollary}[theorem]{Corollary}

\theoremstyle{definition}
\newtheorem{definition}[theorem]{Definition}

\newcommand{\half}{\frac{1}{2}}

\theoremstyle{remark}
\newtheorem{remark}[theorem]{Remark}
\newtheorem*{claim}{Claim}
\newtheorem{example}[theorem]{Example}
\allowdisplaybreaks

\title{Quasi-polynomiality and $N$-point functions of single connected leaky completed Hurwitz numbers}
\author{Chongyu Wang, Chenglang Yang}
\date{}

\begin{document}

\maketitle

\begin{abstract}
    In this paper, we study the structures of single connected $k$-leaky $(r+1)$-completed Hurwitz numbers. We first prove that,
    for fixed genus, after extracting an explicit product of Pochhammer type factors, the stable connected leaky Hurwitz numbers for partition $\mu$ are polynomials in the quotients $[\mu_i]$ modulo $k+r$. We then give a closed formula for the generating function of single connected leaky Hurwitz numbers,
    which can be used to study the genus dependence of leaky Hurwitz numbers.
\end{abstract}

\section{Introduction}

Hurwitz numbers of a nonsingular curve $X$ enumerate covers of $X$ with specified ramification data, which have a long history dating back to Hurwitz \cite{Hurw91,Hurw01}.
They have expressions in terms of the class algebra of the symmetric group, hence are related to representation theory. 
In the past thirty years,
much important progress has been made by mathematicians, which establishes deep relations between various Hurwitz numbers and combinatorics, algebraic geometry, tropical geometry, integrable systems and mathematical physics.
See \cite{CJM,DYZ,ELSV,GJV,KL07,O00,OP06} and references therein.

The famous ELSV formula \cite{ELSV} relates the single connected Hurwitz numbers to the intersection theory of moduli spaces of curves,
which proves the quasi-polynomiality of these Hurwitz numbers (see also \cite{DKOSS,DLPS}). 
The integrability of generating functions of Hurwitz numbers is used to study their relations to Gromov--Witten invariants \cite{O00,OP06} and asymptotic behaviours \cite{DHR,DYZ,Y25}.
The tropical method has also played an important role in studying combinatorial properties of Hurwitz numbers and mirror symmetry \cite{BBBM17,BGM22,CJM,CJMR,CMR,HL}.

Recently, Cavalieri, Markwig and Ranganathan \cite{CMR} made significant progress in studying Hurwitz numbers in terms of logarithmic intersection theory on pluricanonical double ramification cycles.
This motivated them to define the $k$-pluricanonical double Hurwitz numbers and $k$-leaky double Hurwitz numbers.
They studied the tropical interpretation, piecewise polynomiality and Fock space formalism of these numbers.
Later, the authors of \cite{CMSa} generalized the $k$-leaky double Hurwitz numbers to $k$-leaky double Hurwitz descendants and studied several properties of these numbers (see also \cite{CMSb}).
The completed-cycle version of leaky double Hurwitz numbers was introduced in \cite{AKL} using the semi-infinite wedge formalism, and the authors also studied the piecewise polynomiality and the chamber structure.
More recently, Hahn and Kramer \cite{HK} proved the topological recursion for fixed-leakiness Hurwitz numbers associated with a broad class of cut-and-join operators.
In this paper, we focus on the single connected $k$-leaky $(r+1)$-completed Hurwitz numbers,
proving their quasi-polynomiality and deriving a compact $N$-point formula for their generating series.

The first main result of this paper is the following quasi-polynomiality theorem for the stable single connected Hurwitz numbers, proved by direct bosonic computations.
We also refer the reader to the recent work \cite{HK}, whose topological recursion framework may provide a broader structural perspective on the quasi-polynomiality established here.

\begin{theorem}\label{thm:main thm 1}
    The single connected leaky Hurwitz numbers $h_{g,\mu}^{k,r,\circ}$ are quasi-polynomials in $\mu=(\mu_1,\cdots,\mu_N)$ in the stable $2g-2+N>0$ case. More specifically, let $s=\frac{2g-2+N+|\mu|}{k+r}$, when $N>1$,
    \[h_{g,\mu}^{k,r,\circ}=\frac{s!}{|\mathrm{Aut}(\mu)|}\prod_{i=1}^N\frac{((r+1)\mu_i-k\sigma)_{[\mu_i]-\sigma}}{[\mu_i]!((r+1)!)^{[\mu_i]}}P_{\langle \mu_1\rangle,\cdots,\langle \mu_N\rangle}([\mu_1],\cdots,[\mu_N]),\]
    where $|\mathrm{Aut}(\mu)|=\prod_{i=1}^\infty m_i!$ for $\mu=(1^{m_1}2^{m_2}3^{m_3}\cdots)$, $\sigma=2g-2+2N$, $\big((r+1)\mu_i-k\sigma\big)_{[\mu_i]-\sigma}$ is the $k$-Pochhammer symbol given in Definition \ref{def: Pochhammer}, $[n]$ and $\langle n\rangle$ are the quotient and remainder modulo $k+r$,
and $P_{\langle \mu_1\rangle,\cdots,\langle \mu_N\rangle}$ is a polynomial depending on parameters $\langle \mu_1\rangle,\cdots,\langle \mu_N\rangle$.

When $N=1$, $h^{k,r,\circ}_{g,(n)}$ is nonzero only if $n\equiv1-2g\;(\mathrm{mod}\;(k+r))$ and $r[n]_k\geq2g-1+\langle n\rangle_k$. If $r[n]_k>2g-1+\langle n\rangle_k$, $h^{k,r,\circ}_{g,(n)}$ is of the form
        \[\frac{\big((r+1)n-2kg\big)_{s-2g}}{(n-ks)((r+1)!)^s}P(n),\]
        with $P(n)$ a polynomial in $n$.
    
\end{theorem}
\noindent See Remark \ref{rem:criticalcase} for the cases of $N=1$ and $r[n]_k=2g-1+\langle n\rangle_k$ .

The second main result is the following closed formula for the connected leaky Hurwitz numbers.
This formula is based on the KP integrability \cite{DJM,Z15} and provides a sufficient tool to compute the connected leaky Hurwitz numbers and study their properties.
\begin{theorem}\label{thm: main thm 2}
For a given partition $\mu=(\mu_1,\cdots,\mu_N)$, when $N=1$, we have
\begin{align*}
    h^{k,r,\circ}_{g,(n)}
    =\frac{s!}{n}
     \cdot \sum_{ik+t\geq a+\frac{1}{2}, i+p=s \atop t+u=(rn-k(2g-1))/(k+r)}
      \frac{(-1)^{p+u}}{i!p!t!u!((r+1)!)^{i+p}}(-a-\frac{k}{2}+t+ik)_i^{r+1}(-a-\frac{k}{2}-u)_p^{r+1},
\end{align*}
where the sum is over $a\in\mathbb{N}+\half,\,i,p,t,u\in \mathbb{Z}_{\geq0}$.
When $N>1$, we have
\begin{align}
    h^{k,r,\circ}_{g,\mu}
    =\frac{(-1)^{N-1} s!}{\mathfrak{z}(\mu)}[\hbar^{s}\prod_{i=1}^Nz_i^{-\mu_i-1}]\sum_{\text{$N$-cycles } \sigma}
	\prod_{i=1}^N \widehat{A}^{leaky}_{k;r} (z_{\sigma^{i}(1)},z_{\sigma^{i+1}(1)}),
\end{align}
where $[z^a]f(z)$ means taking the coefficient of $z^a$ in $f(z)$,
$\mathfrak{z}(\mu)=|\mathrm{Aut}(\mu)|\prod\mu_i$ and
\begin{align*}
    \widehat{A}_{k;r}^{leaky}(z_i,z_j)=
    \begin{cases}
        \sum_{a\in\mathbb{N}+\half}\mathcal{P}_{-a}(z_j)\,\mathcal{Q}_{-a}(z_i), &\text{\ if\ }i<j,\\
        -\sum_{a\in\mathbb{N}+\half}\mathcal{P}_{a}(z_j)\,\mathcal{Q}_{a}(z_i), &\text{\ if\ }i>j.
    \end{cases}
\end{align*}
The functions $\mathcal{P}_{a}(w)$ and $\mathcal{Q}_{a}(z)$ are explicitly given by
\begin{align}
    \mathcal{P}_a(w)&=w^{-a-\frac12}\!\!\sum_{i,t\geq0}
        \frac{\hbar^{i}\,w^{-ik-t}}{i!\,t!\,((r+1)!)^{i}}
        (a-\frac{k}{2}+t+ik)_i^{r+1},\\
    \mathcal{Q}_b(z)&=z^{b-\frac12}\!\!\sum_{p,t\geq0}
        \frac{(-\hbar)^{p}(-1)^{t}\,z^{-pk-t}}{p!\,t!\,((r+1)!)^{p}}
        (b-\frac{k}{2}-t)_p^{r+1}.
\end{align}
\end{theorem}

As an application of the formula in Theorem \ref{thm: main thm 2}, we study the genus dependence of $k=r=1$ case of the single connected leaky Hurwitz numbers as follows
\begin{theorem}\label{thm:main thm 3}
    For $k=r=1$ and $\mu=(\mu_1,\dots,\mu_N)$, the connected Hurwitz number $h^{1,1,\circ}_{g,\mu}$ is non-zero only if $|\mu|\equiv N\pmod 2$.
    In this case,
    $g_{\max}:=\frac{|\mu|-N}{2}$ is the largest genus for which $h^{1,1,\circ}_{g,\mu}\neq0$ and more precisely,
\begin{align}
    h^{1,1,\circ}_{g_{\max},\mu}
    =\frac{(|\mu|-1)!}{\mathfrak{z}(\mu)}
    \cdot \frac{\sum_{S\subseteq[N]}(-1)^{|S|+\sum_{j\in S}\mu_j} \big(\sum_{j\in S}\mu_j\big)! \big(\sum_{j\in S^c}\mu_j\big)!}{2^{|\mu|-1}(|\mu|+1)}.
\end{align}
\end{theorem}

The structure of this paper is as follows: In Section \ref{sec: prelim}, we review the definition of leaky Hurwitz numbers via tropical geometry, the basics of semi-infinite wedge formalism and connected $N$-point function of a tau-function of the KP hierarchy. In Section \ref{sec: bosons}, we express the generating function of leaky Hurwitz numbers via a recursively defined function $F$, and study its property. Using this, we prove Theorem \ref{thm:main thm 1} in Section \ref{sec:quasipolynomiality}. In Section \ref{sec: conn n point function}, we prove Theorems \ref{thm: main thm 2} and \ref{thm:main thm 3}.

\section*{Acknowledgements}
Chongyu Wang was partially supported by the National Natural Science Foundation of China (No. 12288201).
Chenglang Yang was partially supported by the National Natural Science Foundation of China (No. 12401079) and Natural Science Foundation of Hubei (No. 2026AFA008).

\section{Preliminaries}\label{sec: prelim}

\subsection{Tropical leaky completed Hurwitz numbers}

The main objects we study in this paper count tropical covers of $\mathbf{P}^1$ with prescribed ramification data over 0 and $\infty$ and with prescribed leaky completed-cycle conditions at other ramification points. We review in this subsection basics of tropical geometry and review the definition of tropical leaky completed Hurwitz numbers. We fix throughout this paper two positive integers $k,r>0$.

\begin{definition}\label{def: tropical curve}
    A tropical curve $\Gamma$ (possibly disconnected) is a bounded graph endowed with the following data:
\begin{itemize}
    \item It has a set of vertices $V(\Gamma)$ and a set of half edges $H(\Gamma)$.
    \item There is an incidence map $v$ associating any half edge $h$ to a vertex $v(h)$. The cardinality of $v^{-1}(w)$ is called the valency of a vertex $w\in V(\Gamma)$, and is required to be positive for any $w$. Vertices with valency 1 are called leaves. Vertices which are not leaves are called inner vertices. The set of inner vertices and leaves are denoted by $V^0(\Gamma)$ and $V^{\infty}(\Gamma)$ respectively. 
    \item There is an involution $\iota$ on $H(\Gamma)$ which has no fixed point. Any orbit of $\iota$ is called an edge, we denote the set of edges by $E(\Gamma)$. We require that $v(h)\neq v(\iota(h))$ for any $h\in H(\Gamma)$. We let $n_{\infty}(e)$ be the number of leaves incident to $e\in E(\Gamma)$.
    \item There is a function $g:V(\Gamma)\to \mathbb{N}$ giving the genus of vertices.
\end{itemize}
\end{definition}
\noindent The genus of a tropical curve $\Gamma$ is defined to be 
\[g(\Gamma)=b^1(\Gamma)+\sum_{v\in V^0(\Gamma)}g(v)+1-C(\Gamma),\]
where $b^1(\Gamma)$ is the first Betti number of the geometric realization of $\Gamma$ and $C(\Gamma)$ is the number of its connected components.

Define the tropical projective line $\mathbf{P}^1_{\text{trop},s}$ with $V^0(\mathbf{P}^1_{\text{trop},s})=\{p_1,\cdots,p_s\}$ to be 

\vspace{5pt}

\begin{center}
    \begin{tikzpicture}
    \fill (-1,0) circle (1pt);
    \draw (-1,0) -- (0,0);
    \fill (0,0) circle (1pt);
    \draw (0,0) -- (1,0);
    \fill (1,0) circle (1pt);
    \draw (1,0) -- (1.5,0);
    \fill (1.8,0) circle (0.5pt);
    \fill (1.9,0) circle (0.5pt);
    \fill (2,0) circle (0.5pt);
    \draw (2.3,0) -- (2.8,0);
    \fill (2.8,0) circle (1pt);
    \draw (2.8,0) -- (3.8,0);
    \fill (3.8,0) circle (1pt);
    \node[above left] at (-1,0) {$-\infty$};
    \node[above right] at (3.8,0) {$\infty$};
    \node[above] at (0,0) {$p_1$};
    \node[above] at (1,0) {$p_2$};
    \node[above] at (2.8,0) {$p_s$};
\end{tikzpicture}
\end{center}
with two leaves being $-\infty$ and $\infty$. The right half edge incident to any vertex as in the above picture is defined to be positively directed, and the others are defined to be negatively directed.

\begin{definition}\label{def: tropical cover}
Leaky tropical covers of $\mathbf{P}^1_{\text{trop},s}$ are defined to be maps $\pi$ from the geometric realizations of a tropical curve $\Gamma$ to $\mathbf{P}^1_{\text{trop},s}$, such that
\begin{itemize}
    \item We have 
    
    \[\pi(V^0(\Gamma))=V^0(\mathbf{P}^1_{\text{trop},s}),\;\pi(V^\infty(\Gamma))\subset V^\infty(\mathbf{P}^1_{\text{trop},s}).\]
    And $\pi^{-1}(p_i)$ contains exactly one vertex $v_i\in V^0(\Gamma)$. We denote its valency by $w_i$.
    \item Any edge is associated with a positive integer weight $\omega(e)\in \mathbb{Z}_{>0}$. The weight $\omega(h)$ of a half edge $h$ is defined to be the weight of its corresponding edge.
    \item Any edge $e\in E(\Gamma)$ is mapped injectively into $\mathbf{P}^1_{\text{trop},s}$, hence any half edge has a direction from its image. We let $H^+(\Gamma)$ and $H^-(\Gamma)$ be the set of half edges with positive and negative direction respectively. The number

    \[k(w)=\sum_{\substack{h\in H^+(\Gamma)\\v(h)=w}}\omega(h)-\sum_{\substack{h\in H^-(\Gamma)\\v(h)=w}}\omega(h)\]
    is called the leakiness of a point $w\in V(\Gamma)$.
\end{itemize}
\end{definition}

\noindent An isomorphism between two leaky tropical covers with the same number of inner vertices $\pi_1:\Gamma_1\to\mathbf{P}^1_{\text{trop},p_1,\cdots,p_s}$ and $\pi_2:\Gamma_2\to\mathbf{P}^1_{\text{trop},q_1,\cdots,q_s}$ is the combination of isomorphisms of the source and target curves, denoted $i_{\mathrm{sour}}$ and $i_{\mathrm{tar}}$ respectively, such that they commute with $\pi$ in the following sense
\begin{center}
    \begin{tikzcd}
\Gamma_1 \arrow[r, "i_{\mathrm{sour}}"] \arrow[d, "\pi_1"'] & \Gamma_2 \arrow[d, "\pi_2"] \\
\mathbf{P}^1_{\text{trop},p_1,\cdots,p_s} \arrow[r, "i_{\mathrm{tar}}"']                & \mathbf{P}^1_{\text{trop},q_1,\cdots,q_s}.
\end{tikzcd}
\end{center}
The isomorphisms $i_{\mathrm{sour}}$ and $i_{\mathrm{tar}}$ are required to preserve all structures of $\Gamma$ and $\mathbf{P}^1_{\text{trop},s}$ defined above.

For $s\in \mathbb{N}$ and let $\mu,\nu$ be two partitions. One requires that
\begin{equation}\label{eq: condition on k and r}
    |\nu|+ks=|\mu|.
\end{equation}
The genus is determined by
\begin{equation}\label{eq: RH}
    g=1+\frac{rs-l(\mu)-l(\nu)}{2}.
\end{equation}

\begin{definition}\label{def: tropical cover with parameter}
    We define $C(\nu,\mu,k,r,s)$ to be the set of isomorphism classes of tropical covers such that
    \begin{itemize}
        \item $w_i=r+2-2g(v_i)$ and $k(v_i)=k$ for $1\leq i\leq s$.
        \item The source curve $\Gamma$ has $l(\mu)+l(\nu)$ leaves mapped to $\pm\infty$, whose associated edges are with weights $\nu$ for $\pi^{-1}(-\infty)$ and $\mu$ for $\pi^{-1}(\infty)$.
        \item The genus of $\Gamma$ is $g$ which is defined by equation \eqref{eq: RH}.
    \end{itemize}

\end{definition}

\begin{example}
    The following is a tropical cover in $C(\nu,\mu,1,1,3)$ with $\nu=(1,1)$ and $\mu=(3,1,1)$.

    
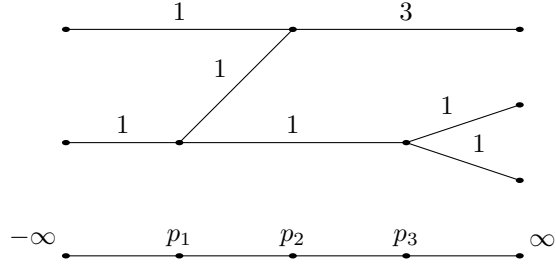
\begin{center}
    \begin{tikzpicture}[xscale=1.5]
    \fill (-1,0) circle (1pt);
    \draw (-1,0) -- (0,0);
    \fill (0,0) circle (1pt);
    \draw (0,0) -- (1,0);
    \fill (1,0) circle (1pt);
    \draw (1,0) -- (2,0);
    \fill (2,0) circle (1pt);
    \draw (2,0) -- (3,0);
    \fill (3,0) circle (1pt);
    \node[above left] at (-1,0) {$-\infty$};
    \node[above right] at (3,0) {$\infty$};
    \node[above] at (0,0) {$p_1$};
    \node[above] at (1,0) {$p_2$};
    \node[above] at (2,0) {$p_3$};
    \fill (-1,1.5) circle (1pt);
    \fill (-1,3) circle (1pt);
    \fill (0,1.5) circle (1pt);
    \fill (1,3) circle (1pt);
    \fill (3,1) circle (1pt);
    \fill (3,2) circle (1pt);
    \fill (3,3) circle (1pt);
    \fill (2,1.5) circle (1pt);
    \draw (-1,1.5) -- (0,1.5);
    \draw (0,1.5) -- (1,3);
    \draw (0,1.5) -- (2,1.5);
    \draw (2,1.5) -- (3,1);
    \draw (2,1.5) -- (3,2);
    \draw (-1,3) -- (1,3);
    \draw (1,3) -- (3,3);
    \node[above] at (0,3) {1};
    \node[above] at (2,3) {3};
    \node[above left] at (0.5,2.25) {1};
    \node[above] at (-0.5,1.5) {1};
    \node[above] at (1,1.5) {1};
    \node[above left] at (2.5,1.75) {1};
    \node[above right] at (2.5,1.25) {1};
\end{tikzpicture}
\end{center}
The inner vertices of the source curve are all with genus 0. The automorphism group of this tropical cover corresponds to commuting 1's in $\mu$, hence has cardinality 2.
\end{example}

Define for any two partitions $\mu$ and $\nu$
\[m_g(\nu,\mu):=[z^{2g}]\frac{\prod_{i=1}^{l(\mu)}\mathcal{S}(\mu_iz)\prod_{j=1}^{l(\nu)}\mathcal{S}(\nu_jz)}{\mathcal{S}(z)},\]
where $\mathcal{S}(z):=\frac{2\mathrm{sinh}(z/2)}{z}$, see \cite{HK}.

\begin{definition}
    We define the tropical $k$-leaky $(r+1)$-completed possibly disconnected double Hurwitz numbers $h^{k,r,s}_{\nu,\mu}$ to be 
    \[\sum_{\pi\in C(\nu,\mu,k,r,s)}\frac{1}{|\mathrm{Aut}(\pi)|}\prod_{i=1}^sm_{g_i}(\textbf{x}_i^+,\textbf{x}^-_i)\prod_{e\in E(\Gamma)}\omega(e)^{1-n_{\infty}(e)},\]
    where $n_\infty(e)$ is defined in Definition \ref{def: tropical curve}, and $\textbf{x}_i^+$ and $\textbf{x}_i^-$ are the set of weights of $h\in H^+(\Gamma)$ and $h\in H^-(\Gamma)$ respectively such that $v(h)=v_i$. We will also denote it by $h^{k,r}_{g,\nu,\mu}$. The connected tropical leaky Hurwitz numbers $h^{k,r,s,\circ}_{\nu,\mu}$ are defined by the same formula with the condition that the geometric realization of the domain curve of $\pi$ is connected.
\end{definition}

\subsection{Semi-infinite wedge formalism}

The main tool we use to study leaky completed Hurwitz numbers is its semi-infinite wedge formalism, so we review in this subsection some knowledge about semi-infinite wedge space, for more details see for example \cite{J15,OP06}. Let $V:=\oplus_{i\in \mathbb{Z}+\frac{1}{2}}\mathbb{C}\underline{i}$ be the infinite dimensional complex vector space spanned by vectors indexed by half integers. The semi-infinite wedge space $\mathcal{V}=\wedge^{\frac{\infty}{2}}V$ is by definition the complex vector space spanned by vectors
\begin{equation}\label{eq: form of element of semi-infinite space}
    \underline{k_1}\wedge\underline{k_2}\wedge\underline{k_3}\wedge\cdots,
\end{equation}
where $k_i\in\mathbb{Z}+\frac{1}{2}$ and $k_i+i-\frac{1}{2}$ is a constant for large $i$. This constant is called the charge of the vector. Let $\mathcal{V}_0$ be the subspace spanned by charge 0 vectors, which are in 1-1 correspondence with the set $\mathcal{P}$ of all partitions: a partition $\mu$ corresponds to $|\mu\rangle:=\underline{\mu_1-\frac{1}{2}}\wedge\underline{\mu_2-\frac{3}{2}}\wedge\cdots$. 
Space $\mathcal{V}$ is assigned with an inner product $(\cdot,\cdot)$ such that the elements \eqref{eq: form of element of semi-infinite space} are orthonormal. We call $|0\rangle:=|\emptyset\rangle$ the vacuum. The covacuum $\langle0|$ is defined to be $(|0\rangle,\cdot)$ in $\mathcal{V}_0^*$.

For any operator $\mathcal{O}$ on $\mathcal{V}_0$, let $\langle \mathcal{O}\rangle:=\langle0|\mathcal{O}|0\rangle$, and call it the vacuum expectation value of $\mathcal{O}$. For any $k\in \mathbb{Z}+\frac{1}{2}$, the fermionic operator $\psi_k$ is defined to be wedge product with $\underline{k}$
\[\psi_k(v)=\underline{k}\wedge v.\]
The operator $\psi^*_k$ is the adjoint of $\psi_k$ with respect to $(\cdot,\cdot)$. Define operators $E_{ij}$ for $i,j\in\mathbb{Z}+\frac{1}{2}$ to be the normal ordering of $\psi_i$ and $\psi_j^*$
\[E_{ij}:=\begin{cases}
    \psi_i\psi_j^*,&j>0\\
    -\psi_j^*\psi_i,&j<0
\end{cases}.\]

We collect some commutation relations:
\begin{equation}\label{eq: comm E psi}
    [E_{ab},\psi_c]=\delta_{bc}\psi_a,\qquad[E_{ab},\psi^*_c]=-\delta_{ac}\psi_b^*,
\end{equation}
and
\begin{align}
        [\sum_{l\in\mathbb{Z}+\frac{1}{2}}g_lE_{l-a,l},\sum_{k\in\mathbb{Z}+\frac{1}{2}}f_k&E_{k-b,k}]=\sum_{l\in\mathbb{Z}+\frac{1}{2}}(g_{l-b}f_l-g_lf_{l-a})E_{l-(a+b),l}\nonumber\\
        &+\delta_{a+b}\delta_{a>0}\sum_{\substack{l\in\mathbb{Z}+\frac{1}{2}\\
        0<l<a}}g_lf_{l-a}+\delta_{a+b}\delta_{b>0}\sum_{\substack{l\in\mathbb{Z}+\frac{1}{2}\\
        0<l<b}}g_{l-b}f_{l}.\label{eq: compute commutator of general operator}
    \end{align}

Let $\varsigma(z):=e^{z/2}-e^{-z/2}=2 \,\text{sinh}(z/2)$. Define 
\[\mathcal{E}_n(z):=\sum_{l\in \mathbb{Z}+\frac{1}{2}}e^{z(l-\frac{n}{2})}E_{l-n,l}+\frac{\delta_{n0}}{\varsigma(z)},\]
for any integer $n$ and formal variable $z$. Let
\[\alpha_n:=\mathcal{E}_n(0)=\sum_{l\in\mathbb{Z}+\frac{1}{2}}E_{l-n,l},\]
for non-zero integer $n$. The commutation relations of these operators are
\[[\mathcal{E}_a(z),\mathcal{E}_b(w)]=\varsigma(aw-bz)\mathcal{E}_{a+b}(z+w),\;\;\;[\alpha_n,\alpha_m]=n\delta_{n+m},\]
see \cite{OP06}. For positive integer $r$, define 
\[\mathcal{F}^{(r)}:=\sum_{l\in\mathbb{Z}+\frac{1}{2}}\frac{l^r}{r!}E_{l,l}.\]
We call $E:=\mathcal{F}^{(1)}$ the energy operator. If an operator $\mathcal{O}$ satisfies $[\mathcal{O},E]=e\mathcal{O}$, we say it has energy $e$. The operators $\mathcal{E}_n(z)$ and hence $\alpha_n$ have energy $n$. An operator with positive energy annihilates the vacuum $|0\rangle$. Similarly, an operator with negative energy annihilates the covacuum $\langle0|$.

For a partition $\mu$, we let 
\begin{equation}\label{eq: def zmu}
    \mathfrak{z}(\mu)=|\mathrm{Aut}(\mu)|\prod_{i=1}^{l(\mu)}\mu_i,
\end{equation}
where $\mathrm{Aut}(\mu)$ is the symmetry group permuting equal parts of $\mu$. Leaky Hurwitz numbers defined in the last subsection can be computed by \cite{AKL,CMR,HK}
\begin{equation}\label{eq: semi infinite wedge for hur}
    h^{k,r,s}_{\nu,\mu}=h^{k,r}_{g,\nu,\mu}=[z_1^{r+1}z_2^{r+1}\cdots z_s^{r+1}]\frac{1}{\mathfrak{z}(\mu)\mathfrak{z}(\nu)}\langle \prod_{j=1}^{l(\nu)}\alpha_{\nu_j}\prod_{p=1}^s\mathcal{E}_k(z_p)\prod_{i=1}^{l(\mu)}\alpha_{-\mu_i}\rangle,
\end{equation}

\noindent where $[z_1^{r+1}\cdots z_s^{r+1}]$ means selecting the coefficient of the monomial $z_1^{r+1}\cdots z_s^{r+1}$, and $g,s$ are related via equation \eqref{eq: RH}. Note that it differs from Definition 3.1 in \cite{AKL} by a factor $1/|\mathrm{Aut}(\mu)||\mathrm{Aut}(\nu)|$, and the convention of leakiness is reversed. 

By letting $\nu=(1^{l(\nu)})$, single $k$-leaky $(r+1)$-completed cycles Hurwitz numbers are then
\begin{equation}\label{eq: def single hur}
    h^{k,r}_{g,\mu}=h^{k,r,s}_\mu=[z_1^{r+1}\cdots z_s^{r+1}]\frac{1}{l(\nu)!\mathfrak{z}(\mu)}\langle \alpha_1^{l(\nu)}\prod_{p=1}^s\mathcal{E}_{k}(z_p)\prod_{i=1}^{l(\mu)}\alpha_{-\mu_i}\rangle.
\end{equation}
Then equation \eqref{eq: condition on k and r} now becomes
\begin{equation}\label{eq: total energy 0}
    l(\nu)+ks=|\mu|,
\end{equation}
which together with equation \eqref{eq: RH} implies
\begin{equation}\label{eq: relation of g and s}
    (k+r)s=2g-2+l(\mu)+|\mu|.
\end{equation}

Define the operator 
\[\mathcal{H}_{k,r}:=[z^{r+1}]\mathcal{E}_k(z)=\sum_{l\in \mathbb{Z}+\frac{1}{2}}\frac{\hat{l}^{r'}}{r'!}E_{l-k,l},\]
where
\begin{equation}\label{eq:definelhatrprime}
    \hat{l}:=l-\frac{k}{2},\qquad r':=r+1.
\end{equation}
The generating series of single leaky completed Hurwitz numbers is defined to be
\begin{equation}\label{eq: sym in nu}
    h^{k,r}_\mu(\hbar):=\sum_{s\geq 0}\frac{1}{s!}\hbar^sh^{k,r,s}_\mu=\frac{1}{\mathfrak{z}(\mu)}\langle e^{\alpha_1}e^{\hbar\mathcal{H}_{k,r}}\prod_{i=1}^{l(\mu)}\alpha_{-\mu_i}\rangle.
\end{equation}
Let $\mathbf{t}:=(t_1,t_2,\cdots)$. We further package formula \eqref{eq: sym in nu} to
\begin{equation}\label{eq: partition function}
    \tau^{leaky}_{k,r}(\mathbf{t},\hbar):=\sum_{\mu\in\mathcal{P}}h^{k,r}_\mu(\hbar)\prod_{i=1}^{l(\mu)}\mu_it_{\mu_i}=\langle  e^{\alpha_1} e^{\hbar \mathcal{H}_{k;r}} \Gamma_-(\mathbf{t}) \rangle,
\end{equation}
which is called the partition function of $k$-leaky $(r+1)$-completed Hurwitz numbers. Here $\Gamma_-(\mathbf{t}):=\exp\big(\sum_{j=1}^\infty t_j\alpha_{-j}\big)$.
Note that since the operator $e^{\alpha_1} e^{\hbar \mathcal{H}_{k;r}}\in\widehat{GL}(\infty)$ preserves the KP integrability (see, for example, \cite{DJM}),
the generating function of leaky Hurwitz numbers given by formula \eqref{eq: partition function} satisfies the KP hierarchy.

\subsection{Connected $N$-point function for tau-functions of the KP hierarchy}
There is a general method deriving the connected $N$-point function from a tau-function of the KP hierarchy. We introduce the related notions in this subsection, and these are used in Section \ref{sec: conn n point function}.

A formal power series $\tau(\mathbf{t})\in\mathbb{C}[\![\mathbf{t}]\!]$ is a tau-function of the KP hierarchy if it satisfies the following Hirota bilinear equation (see \cite{DJM} for more details)
\begin{align}\label{eqn:hirota}
	\oint\ e^{\xi(\mathbf{t}-\mathbf{t}',z)}
	\tau(\mathbf{t}-[z^{-1}]) \tau(\mathbf{t}'+[z^{-1}])\ dz=0,
\end{align}
where $\xi(\mathbf{t},z):=\sum_{k=1}^\infty t_k z^k$, and
\begin{align*}
	\mathbf{t}\pm[z^{-1}]
	=(t_1\pm z^{-1}, t_2\pm\frac{1}{2}z^{-2},
	...,t_n\pm\frac{1}{n}z^{-n},...).
\end{align*}
Since $\{s_\mu(\mathbf{t})\,|\,\mu\in\mathcal{P}\}$ forms a topological basis of $\mathbb{C}[\![\mathbf{t}]\!]$, for any tau-function of the KP hierarchy $\tau_{KP}(\mathbf{t})$,
we can expand it in terms of Schur polynomials as
\begin{align}\label{eqn:tauKP as slambda}
    \tau_{KP}(\mathbf{t})
    =\sum_{\mu\in\mathcal{P}} c_{\mu} \cdot s_{\mu}(\mathbf{t}),
\end{align}
where $c_{\mu}\in\mathbb{C}$ and $s_{\mu}(\mathbf{t})$ is the Schur polynomial.
Without loss of generality,
we can assume $c_{\emptyset}=\tau_{KP}(0)=1$.
Then the Hirota bilinear relation \eqref{eqn:hirota} satisfied by the KP tau-function $\tau_{KP}(\mathbf{t})$ is equivalent to the following relations between the coefficients $\{c_{\lambda}\}_{\lambda\in\mathcal{P}}$:
\begin{align}\label{eqn:clambda by c}
    c_{\mu} = \det\big(c_{(m_i|n_j)}\big)_{1\leq i,j\leq l},
\end{align}
where $\mu=(m_1,\cdots,m_l|n_1,\cdots,n_l)$ is the Frobenius notation.
As a consequence,
the tau-function $\tau_{KP}(\mathbf{t})$ is determined by $\{c_{(a|b)}\}_{a,b\in\mathbb{N}}$.
For convenience,
we denote the generating function of these numbers $\{c_{(a|b)}\}_{a,b\in\mathbb{N}}$ by (see \cite{Z15})
\begin{align}\label{eqn:def Azw}
    A(z,w)=\sum_{n,m\in\mathbb{N}} (-1)^n c_{(m|n)} \cdot z^{-n-1} w^{-m-1}.
\end{align}
This generating function is also called the fermionic two-point function in the literature.

From equations \eqref{eqn:tauKP as slambda} and \eqref{eqn:clambda by c},
the generating function $A(z,w)$ uniquely specifies the KP tau-function $\tau_{KP}(\mathbf{t})$ with the condition $\tau_{KP}(0)=1$.
The following formula for connected $N$-point function of KP tau-functions provides a direct way to compute the connected correlators.
We recommend the references \cite{BE12,TW94,Z15}.
The result is,
for the $N=1$ case (see equation (211) in \cite{Z15}),
\begin{align}\label{eqn:conn n=1}
	\sum_{j\geq1}
	\bigg(\frac{\partial \log \tau_{KP}(\mathbf{t})}{\partial t_{j}} \Big|_{\mathbf{t}=0}
	\cdot z^{-j-1}\bigg)
	=\lim_{w\rightarrow z} A(z,w),
\end{align}
and for general $N\geq 2$ (see Theorem 5.3 in \cite{Z15}),
\begin{align}\label{eqn:conn n}
	\begin{split}
	&\sum_{j_1,\cdots,j_N\geq 1}
	\bigg(\frac{\partial^N \log \tau_{KP}(\mathbf{t})}{\partial t_{j_1} \cdots \partial t_{j_N}} \Big|_{\mathbf{t}=0}
	\cdot z_1^{-j_1-1} \cdots z_N^{-j_N-1}\bigg) \\
	=&(-1)^{N-1}\cdot \sum_{\text{$N$-cycles }: \sigma}
	\prod_{i=1}^N \widehat A (z_{\sigma^{i}(1)},z_{\sigma^{i+1}(1)})
	-\delta_{N,2}\cdot\sum_{h\geq0}(h+1)z_1^{-2-h}z_2^h,
	\end{split}
\end{align}
where the sum runs over $N$-cycles $\sigma\in S_N$ and
\begin{align}\label{eq-AandAhat}
	\widehat A(z_i,z_j) =
	\begin{cases}
		\sum_{h\geq 0} z_i^{-1-h}z_j^h + A(z_i,z_j),\qquad
		& i<j,\\
		-\sum_{h\geq 0} z_j^{-1-h}z_i^h + A(z_i,z_j), &i>j.
	\end{cases}
\end{align}

At the end of the preliminaries, we fix some notation

\begin{definition}\label{def: Pochhammer}
    We define $k$-Pochhammer symbol for any $n\text{ and }i\in \mathbb{Z}$ to be

    \[(n)_{i}:=\begin{cases}
        \prod_{j=0}^{i-1}(n-jk)&i>0\\
        1&i=0\\
        \frac{1}{\prod_{j=1}^{-i}(n+jk)}&i<0
    \end{cases}.\]
    And we use the symbol $(n)_{\overline{i}}$ to denote the $(-1)$-Pochhammer symbol. We define $n=i[n]_i+\langle n\rangle_i$ for the integral division of an integer $n$ by a positive integer $i$, where $0\leq \langle n\rangle_i<i$. If $i=k+r$, we will omit the subscript.
    \end{definition}

\section{Leaky Hurwitz numbers via bosons and its property}\label{sec: bosons}
We start from the formula \eqref{eq: sym in nu} to compute the single leaky Hurwitz numbers via bosons.
Note that positive energy operators $\alpha_1$ and $\mathcal{H}_{k,r}$ for $k>0$ annihilate the vacuum, hence 
\[e^{-u\mathcal{H}_{k,r}}e^{-\alpha_1}|0\rangle=|0\rangle.\]
Thus formula \eqref{eq: sym in nu} is equal to 
\begin{equation}\label{eq: hurwitz as prod of A operator}
    h^{k,r}_\mu(\hbar)=\frac{1}{\mathfrak{z}(\mu)}\langle \prod_{i=1}^N(e^{\alpha_1}e^{\hbar\mathcal{H}_{k,r}}\alpha_{-\mu_i}e^{-\hbar\mathcal{H}_{k,r}}e^{-\alpha_1})\rangle,
\end{equation}
where $N:=l(\mu)$. By calculating the conjugate 
\begin{equation}\label{eq: conjugate}
    e^{\alpha_1}e^{u\mathcal{H}_{k,r}}\alpha_{-n}e^{-u\mathcal{H}_{k,r}}e^{-\alpha_1},
\end{equation}
we give in this section an expression of $h^{k,r}_\mu(\hbar)$ via a recursively defined function $F$, and study this function's property.

\subsection{Calculating the conjugate}

As said above, we will give expression for formula \eqref{eq: hurwitz as prod of A operator} by calculating the conjugate \eqref{eq: conjugate}, which consists of first conjugate $\alpha_n$ by $e^{u\mathcal{H}_{k,r}}$ and then conjugate by $e^{\alpha_1}$. The Baker--Campbell--Hausdorff formula reads 
\begin{equation}\label{eq: BCH}
    e^{\hbar\mathcal{H}_{k,r}}\alpha_{-n}e^{-\hbar\mathcal{H}_{k,r}}=\sum_{i=0}^{\infty}\frac{\hbar^i}{i!}\text{ad}_{\mathcal{H}_{k,r}}^i(\alpha_{-n}),
\end{equation}
where $\text{ad}_{U}(V):=[U,V]$ and $\text{ad}_{\mathcal{H}_{k,r}}^i$ means $\text{ad}_{\mathcal{H}_{k,r}}$ acting $i$ times on $\alpha_{-n}$. We inductively define a set of functions which appear in the expression of \eqref{eq: BCH}.

\begin{definition}\label{def: poly for one commutator}
    Functions $F_n^{(i)}(l)$ for $i\geq 0,\,n\neq0$ and $l\in\mathbb{Z}+\frac{1}{2}$ are inductively defined by letting 
    \[F_n^{(0)}(l):=1,\]

    \noindent and 
    \begin{align}
        F_n^{(i+1)}(l)&=\frac{(\hat{l}+n-ki)^{r'}}{r'!}F_n^{(i)}(l)-\frac{\hat{l}^{r'}}{r'!}F_n^{(i)}(l-k).\label{eq: induction for F}
    \end{align}
\end{definition}
\noindent Note that $\hat{l}$ and $r'$ are defined in \eqref{eq:definelhatrprime}, and we omit $k$ and $r$ in the notation of $F$ to simplify our formula. We will also do this for $G$ defined in Proposition \ref{prop: def of G}.

\begin{proposition}\label{prop: first conjugation}
    We have 
    \begin{align}
        e^{\hbar\mathcal{H}_{k,r}}\alpha_{-n}e^{-\hbar\mathcal{H}_{k,r}}
        =&\sum_{i=0}^{\infty}\frac{\hbar^i}{i!}\sum_{l\in \mathbb{Z}+\frac{1}{2}}F_{n}^{(i)}(l)E_{l+n-ki,l}+\frac{\hbar^{[n]_k}\delta_{\langle n\rangle_k}}{[n]_k!r'!}\sum_{\substack{l\in \mathbb{Z}+\frac{1}{2}\\
        0<l<k}}\hat{l}^{r'}F_{n}^{([n]_k-1)}(l-k).\label{eq:firstconjugate}
    \end{align}
\end{proposition}

\begin{proof}
    By equation \eqref{eq: BCH}, we only need to calculate $\text{ad}^i_{\mathcal{H}_{k,r}}(\alpha_{-n})$. We use induction to show that
    \begin{equation}\label{eq: ind to prove adHi}
        \text{ad}^i_{\mathcal{H}_{k,r}}(\alpha_{-n})=\sum_{l\in \mathbb{Z}+\frac{1}{2}}F_{n}^{(i)}(l)E_{l+n-ki,l}+\text{scalar multiplication term}.
    \end{equation}
    For $i=0$, this follows from
    \[\alpha_{-n}=\sum_{l\in\mathbb{Z}+\frac{1}{2}}E_{l+n,l}=\sum_{l\in\mathbb{Z}+\frac{1}{2}}F_n^{(0)}(l)E_{l+n,l}.\]
    Now suppose equation \eqref{eq: ind to prove adHi} holds for $i$, then since commutator of any operator with scalar multiplication is 0
    \begin{align*}
        \text{ad}^{i+1}_{\mathcal{H}_{k,r}}(\alpha_{-n})=&[\sum_{l\in \mathbb{Z}+\frac{1}{2}}\frac{\hat{l}^{r'}}{r'!}E_{l-k,l},\sum_{l\in \mathbb{Z}+\frac{1}{2}}F_{n}^{(i)}(l)E_{l+n-ki,l}]\\
        =&\sum_{l\in \mathbb{Z}+\frac{1}{2}}(\frac{(\hat{l}+n-ki)^{r'}}{r'!}F_{n}^{(i)}(l)-\frac{\hat{l}^{r'}}{r'!}F_{n}^{(i)}(l-k))E_{l+n-(i+1)k,l}\\
        &+\text{scalar multiplication term}\\
        =&\sum_{l\in \mathbb{Z}+\frac{1}{2}}F_{n}^{(i+1)}(l)E_{l+n-(i+1)k,l}+\text{scalar multiplication term}.
    \end{align*}
    The second equality follows from equation \eqref{eq: compute commutator of general operator}, and the last equality follows from equation \eqref{eq: induction for F}. So the induction procedure verifies the non-scalar multiplication term in the proposition. For the scalar multiplication term, equation \eqref{eq: compute commutator of general operator} tells us that it occurs only if $k\,|\,n$ and $i=[n]_k$, and in this case, it is
    \[\sum_{\substack{l\in \mathbb{Z}+\frac{1}{2}\\
        0<l<k}}\frac{\hat{l}^{r'}}{r'!}F_{n}^{([n]_k-1)}(l-k).\]
    The proposition is thus proved.
\end{proof}

Define for any function in $l$ the backward difference operator
\[\nabla f(l):=f(l)-f(l-1).\]

\begin{proposition}\label{prop: second conjugation}
    We have 
    \begin{equation}\label{eq: second conjugate}
    \begin{split}
&e^{\alpha_1}e^{\hbar\mathcal{H}_{k,r}}\alpha_{-n}e^{-\hbar\mathcal{H}_{k,r}}e^{-\alpha_1}\\=&\sum_{i=0}^{\infty}\sum_{t=0}^{\infty}\frac{\hbar^i}{i!t!}\sum_{l\in \mathbb{Z}+\frac{1}{2}}\nabla^t F_n^{(i)}(l)E_{l+n-ki-t,l}
        +\frac{\hbar^{[n]_k}\delta_{\langle n\rangle_k}}{[n]_k!r'!}\sum_{\substack{l\in \mathbb{Z}+\frac{1}{2}\\
        0<l<k}}\hat{l}^{r'}F_{n}^{([n]_k-1)}(l-k)\\
        &+(1-\delta_{\langle n\rangle_k})\frac{\hbar^{[n]_k}}{[n]_k!\langle n\rangle_k!}\nabla^{\langle n\rangle_k-1} F_n^{([n]_k)}(-\frac{1}{2})+\sum_{i=0}^{[n]_k-1}\frac{\hbar^i}{i!(n-ki)!}\nabla^{n-ki-1} F_n^{(i)}(-\frac{1}{2}).
    \end{split}
    \end{equation}
\end{proposition}

\begin{proof}
    Since the conjugation of the scalar multiplication term in equation \eqref{eq:firstconjugate} by $e^{\alpha_1}$ is itself, we only need to compute
    \[e^{\alpha_1}(\sum_{i=0}^{\infty}\frac{\hbar^i}{i!}\sum_{l\in \mathbb{Z}+\frac{1}{2}}F_{n}^{(i)}(l)E_{l+n-ki,l})e^{-\alpha_1}\]
    and prove that it is equal to the sum of the first, third and fourth terms on the right hand side of equation \eqref{eq: second conjugate}.

    We use induction to show that
    \begin{align}
        &\text{ad}_{\alpha_1}^t(\sum_{l\in \mathbb{Z}+\frac{1}{2}}F_{n}^{(i)}(l)E_{l+n-ki,l})=\sum_{l\in \mathbb{Z}+\frac{1}{2}}\nabla^t F_n^{(i)}(l)E_{l+n-ki-t,l}+\text{scalar multiplication term}.\label{eq: adt}
    \end{align}
    For $t=0$ this is true. Suppose it is true for $t$, then since $\alpha_1=\sum_{l\in\mathbb{Z}+\frac{1}{2}}E_{l-1,l}$,
    \begin{align*}
        &\text{ad}_{\alpha_1}(\sum_{l\in \mathbb{Z}+\frac{1}{2}}\nabla^t F_n^{(i)}(l)E_{l+n-ki-t,l})\\
        =&\sum_{l\in \mathbb{Z}+\frac{1}{2}}(\nabla^t F_n^{(i)}(l)-\nabla^t F_n^{(i)}(l-1))E_{l+n-ki-t-1,l}+\text{scalar multiplication term}\\
        =&\sum_{l\in \mathbb{Z}+\frac{1}{2}}\nabla^{t+1} F_n^{(i)}(l)E_{l+n-ki-t-1,l}+\text{scalar multiplication term}.
    \end{align*}
    Thus induction shows that equation \eqref{eq: adt} is true.

    We now need to compute the scalar multiplication term. If $\langle n\rangle_k\neq 0$, the commutator
    \[\text{ad}_{\alpha_1}(\sum_{l\in \mathbb{Z}+\frac{1}{2}}\nabla^{\langle n\rangle_k-1} F_n^{([n]_k)}(l)E_{l+1,l})\]
    produces the third term on the right hand side of equation \eqref{eq: second conjugate}. And the commutator
    \[\text{ad}_{\alpha_1}(\sum_{i=0}^{[n]_k-1}\sum_{l\in \mathbb{Z}+\frac{1}{2}}\nabla^{n-ki-1} F_n^{(i)}(l)E_{l+1,l})\]
    produces the last term on the right hand side of equation \eqref{eq: second conjugate}. The proposition is proved.
\end{proof}

\subsection{Dependence of function $F$ on parameters}

We study in this subsection the dependence of $F_n^{(i)}(l)$ on parameters $i$ and $n$. We first take out the leading coefficient of $F_{n}^{(i)}(l)$.

\begin{proposition}\label{prop: def of G}
    Define 
    \[G_n^{(i)}(l):=\frac{(r'!)^i}{(r'n)_i}F_{n}^{(i)}(l).\]
    Recall that $(\cdots)_i$ is the $k$-Pochhammer symbol. Then they satisfy the following recursive relation
    \begin{equation}\label{eq: recurrence for G}
        G_n^{(i+1)}(l)=\frac{(\hat{l}+n-ki)^{r'}}{r'n-ki}G_n^{(i)}(l)-\frac{\hat{l}^{r'}}{r'n-ki}G_n^{(i)}(l-k).
    \end{equation}
\end{proposition}

\begin{proof}
    Substituting $G$ for $F$ in equation \eqref{eq: induction for F} we get
    \begin{align*}
        \frac{(r'n)_{i+1}}{(r'!)^{i+1}}G_n^{(i+1)}(l)=\frac{(\hat{l}+n-ki)^{r'}}{r'!}\frac{(r'n)_i}{(r'!)^i}G_n^{(i)}(l)-\frac{\hat{l}^{r'}}{r'!}\frac{(r'n)_i}{(r'!)^i}G_n^{(i)}(l-k).
    \end{align*}
    This simplifies to equation \eqref{eq: recurrence for G} because 
    \[(r'n)_{i+1}=(r'n-ki)(r'n)_i.\]
\end{proof}

\begin{lemma}\label{lem: upper bound degree G l}
    $G_n^{(i)}(l)$ is a monic polynomial in $l$ of degree $ir$.
\end{lemma}

\begin{proof}
    We can see this by induction on $i$. For $i=0$, $G_n^{(0)}(l)=F_n^{(0)}(l)=1$. Suppose $G_n^{(i)}(l)$ is a monic polynomial in $l$ of degree $ir$:   
    \[G_n^{(i)}(l)=l^{ir}+al^{ir-1}+\text{lower degree terms}.\]
    Then the degree of $G_n^{(i+1)}(l)$ is not greater than $ir+r+1$ since $\frac{(\hat{l}+n-ki)^{r'}}{r'n-ki}$ and $\frac{\hat{l}^{r'}}{r'n-ki}$ are polynomials in $l$ of degree $r+1$. The coefficient of $l^{(i+1)r+1}$ is $1-1=0$, and the coefficient of $l^{(i+1)r}$ is 
    \[\frac{r'(n-ki-\frac{k}{2})-(-\frac{k}{2})r'-(-k)ir+a-a}{r'n-ki}\]
    which is equal to 1. The lemma is thus proved.
\end{proof}

 We now introduce a basis for the polynomial ring $\mathbb{Q}[l]$, namely
    \begin{equation}\label{def: shifted basis}
        \{\hat{l\vphantom{}}^{\underline{i}}:=(\hat{l}-[i]_rk)^{\langle i\rangle_r}\cdot\prod_{j=0}^{[i]_r-1}(\hat{l}-jk)^r\}.
    \end{equation}
    We will expand $G_n^{(i)}(l)$ with respect to this basis. We give a lemma to relate the basis $\{\hat{l\vphantom{}}^{\underline{ir-j}}\;|\;j\geq0\}$ with $\{l^{ir-j}\;|\;j\geq0\}$.

\begin{lemma}\label{lem: base change between polynomial basis}
    For any fixed $0\leq j\leq s$, when $\hat{l\vphantom{}}^{\underline{ir-j}}$ is written as linear combination of $\{l^{ir-t}\;|\;t\geq j\}$ as polynomials in $l$, the coefficient of $l^{ir-s}$ is a polynomial in $i$, which depends on $k,r,j,s$.
\end{lemma}

\begin{proof}
    We first show the case of expansion with respect to $\{\hat{l}^{ir-t}\;|\;t\geq j\}$. We recall that
    \begin{equation}\label{eq: base change by stirling number}
        (\hat{l})_i=\sum_{t=0}^i(-k)^t {i\brack i-t}\hat{l}^{i-t},
    \end{equation}
    where $i\brack i-t$ is the first kind Stirling number. This can be obtained for example by letting $x=\frac{\hat{l}}{k}$ in equation (6.13) of \cite{GKP}.

    \vspace{10pt}

    \begin{claim}
        The ${i \brack i-t}$ is a polynomial in $i$ for fixed $t$.
    \end{claim}

    \vspace{5pt}

    \noindent \emph{Proof of the Claim}.
    We use induction on $t$ to prove the claim. For $t=0$, the claim is true since ${i \brack i}=1$. Suppose the claim is true for $t$. Since the Stirling numbers satisfy the recursive relation (see page 264 in \cite{GKP})
    \[{i \brack i-t-1}-{i-1 \brack i-t-2}=(i-1){i-1 \brack i-t-1},\]
    we have 
    \[{i \brack i-t-1}={t+1 \brack 0}+\sum_{u=t+1}^{i-1}u{u\brack u-t}=\sum_{u=t+1}^{i-1}u{u\brack u-t}.\]
    By induction hypothesis this is a summation of a polynomial in $u$ hence is a polynomial in $i$ and $t$. Since $t$ is fixed, we have thus proved the claim.

    \vspace{10pt}

    We return to the proof of the lemma. Since 
    \begin{align*}
        [ir-j]_r&=i-[j]_r-1+\delta_{\langle j\rangle_r},\\
        \langle ir-j\rangle_r&=r-\langle j\rangle_r-r\delta_{\langle j\rangle_r},
    \end{align*}
    it can be seen that
    \begin{equation}\label{eq: decompose hatl exponent}
        \hat{l\vphantom{}}^{\underline{ir-j}}=((\hat{l})_{i-[j]_r-1})^r\cdot (\hat{l}-ik+[j]_rk+k)^{r-\langle j\rangle_r},
    \end{equation}
    no matter $\langle j\rangle_r$ vanishes or not. By equation \eqref{eq: base change by stirling number} and the claim, $(\hat{l})_{i-[j]_r-1}$, when written in basis $\hat{l}^{i-t}$, has polynomial coefficient in $i$ for fixed $t$. Then $\hat{l}^{ir-s}$ comes from that of the form
    \[\hat{l}^{|\mathbf{t}|-s}\prod_{a=1}^r\hat{l}^{i-t_a},\]
    where $\mathbf{|t|}\leq s+r-\langle j\rangle_r$ and $t_a\geq[j]_r+1$ for any $1\leq a\leq r$. Here $\hat{l}^{i-t_a}$ comes from $(\hat{l})_{i-[j]_r-1}$ and $\hat{l}^{|\mathbf{t}|-s}$ comes from $(\hat{l}-ik+[j]_rk+k)^{r-\langle j\rangle_r}$. It is a finite sum of products of polynomials in $i$, with coefficients depending on specific values of $k,r,j,s$.
    
    At last we change the basis $\{\hat{l}^{ir-t}\;|\;t\geq j\}$ into $\{l^{ir-t}\;|\;t\geq j\}$. Since the coefficient of $l^{ir-s}$ in $\hat{l}^{ir-t}=(l-\frac{k}{2})^{ir-t}$ is $\binom{ir-t}{s-t}\cdot(-\frac{k}{2})^{s-t}$, it is still a polynomial in $i$. As a consequence the coefficient of $l^{ir-s}$ in $\hat{l\vphantom{}}^{\underline{ir-j}}$ is a sum of finite polynomials hence is still a polynomial in $i$, with coefficients depending on $k,r,j,s$. The lemma is thus proved.
\end{proof}

By Lemma \ref{lem: upper bound degree G l}, we let
\begin{equation}\label{eq: expansion of G}
    G_n^{(i)}(l)=\sum_{j=0}^{ir}g_j(i,n)\hat{l\vphantom{}}^{\underline{ir-j}},
\end{equation}
with $g_0(i,n)=1$. And we simply let $g_j(i,n)=0$ for $j> ir$. 

\begin{lemma}\label{lem: induction for g}
    The coefficients $g_s(i,n)$ satisfy the following recurrence relation 
    \begin{equation}\label{eq: induction for g}
        g_s(i+1,n)=\frac{r'n-ki-ks}{r'n-ki}g_s(i,n)+\sum_{j=s-r}^{s-1}\frac{C_{s,j}(n)}{r'n-ki}g_j(i,n)
    \end{equation}
    with $g_0(i,n)=1$. Here $C_{s,j}(n)$ are polynomials in $n$ with coefficients depending on $k,r,s,j$. In particular, $C_{s,j}(n)$ do not depend on $i$.
\end{lemma}

\begin{proof}
    Notice that $\hat{l\vphantom{}}^{\underline{i+r}}=\hat{l}^r\cdot \widehat{(l-k)}^{\underline{i}}$, hence equation \eqref{eq: recurrence for G} is just
    \begin{align*}
&\sum_{j=0}^{(i+1)r}g_j(i+1,n)\hat{l\vphantom{}}^{\underline{(i+1)r-j}}=\frac{(\hat{l}+n-ki)^{r+1}}{r'n-ki}\sum_{j=0}^{ir}g_j(i,n)\hat{l\vphantom{}}^{\underline{ir-j}}-\frac{\hat{l}}{r'n-ki}\sum_{j=0}^{ir}g_j(i,n)\hat{l\vphantom{}}^{\underline{(i+1)r-j}}.
    \end{align*}

    We first assume $s\leq ir$. The coefficient of $\hat{l\vphantom{}}^{\underline{(i+1)r-s}}$ on the left hand side is $g_s(i+1,n)$. To compute the coefficient of $\hat{l\vphantom{}}^{\underline{(i+1)r-s}}$ on the right hand side, we first exclude terms that give no contribution. For the first summation on the right hand side, terms from $j=0$ to $s-r-1$ have degree in $l$ greater than $(i+1)r-s$, hence give no contribution. Still for the first summation, terms from $j=s+2$ to $ir$ have degree smaller than $(i+1)r-s$, hence still give no contribution. For the second summation, similar argument shows that terms from $j=0$ to $s-1$ and from $j=s+2$ to $ir$ give no contribution. 
    
    The contribution of 
\[\frac{(\hat{l}+n-ki)^{r+1}}{r'n-ki}g_{s+1}(i,n)\hat{l\vphantom{}}^{\underline{ir-s-1}}\]
will cancel out with that of 
\[-\frac{\hat{l}}{r'n-ki}g_{s+1}(i,n)\hat{l\vphantom{}}^{\underline{(i+1)r-s-1}}.\]
So the coefficient of $\hat{l\vphantom{}}^{\underline{(i+1)r-s}}$ on the right hand side comes from
\begin{equation}\label{form: total contribution to l (i+1)r-s}
    \frac{(\hat{l}+n-ki)^{r'}}{r'n-ki}\sum_{j=s-r}^{s}g_j(i,n)\hat{l\vphantom{}}^{\underline{ir-j}}-\frac{\hat{l}}{r'n-ki}g_s(i,n)\hat{l\vphantom{}}^{\underline{(i+1)r-s}}.
\end{equation}

We first calculate contributions containing $g_s(i,n)$, which are the second part and the $j=s$ term of the first part of formula \eqref{form: total contribution to l (i+1)r-s}. Since $\hat{l\vphantom{}}^{\underline{(i+1)r-s+1}}/\hat{l\vphantom{}}^{\underline{(i+1)r-s}}=\hat{l}-k(i-[s]_r+\delta_{\langle s\rangle_r})$, the contribution to $\hat{l\vphantom{}}^{\underline{(i+1)r-s}}$ of the second part is
\begin{equation}\label{eq:firstpart}
    -\frac{k(i-[s]_r+\delta_{\langle s\rangle_r})}{r'n-ki}g_s(i,n)\hat{l\vphantom{}}^{\underline{(i+1)r-s}}.
\end{equation}
We now compute the contribution of the first part of \eqref{form: total contribution to l (i+1)r-s} for $j=s$. Since
\begin{align}
    &\hat{l\vphantom{}}^{\underline{(i+1)r-s+1}}/\hat{l\vphantom{}}^{\underline{ir-s}}=(\hat{l}-k(i-1-[s]_r))^{\langle s\rangle_r}\cdot (\hat{l}-k(i-[s]_r))^{r-\langle s\rangle_r}\cdot(\hat{l}-k(i-[s]_r+\delta_{\langle s\rangle_r})),\label{eq:plus1coeff}
\end{align}
the coefficient of $\hat{l\vphantom{}}^{\underline{(i+1)r-s}}$ is equal to the degree $r$ part of $l$ in $(\hat{l}+n-ki)^{r'}$ minus degree $r$ part of $l$ in \eqref{eq:plus1coeff}, hence it is 
\begin{equation}\label{eq:secondpart}
    \frac{r'(n-ki)+rk(i-[s]_r)-k\langle s\rangle_r+k(i-[s]_r+\delta_{\langle s\rangle_r})}{r'n-ki}g_s(i,n)\hat{l\vphantom{}}^{\underline{(i+1)r-s}}.
\end{equation}
Combining \eqref{eq:firstpart} and \eqref{eq:secondpart}, the total coefficient of $\hat{l\vphantom{}}^{\underline{(i+1)r-s}}$ containing $g_s(i,n)$ is 
\[\frac{r'n-ki-ks}{r'n-ki}g_s(i,n)\hat{l\vphantom{}}^{\underline{(i+1)r-s}}.\]
This is the first part of equation \eqref{eq: induction for g}.

We then study the remaining part of formula \eqref{form: total contribution to l (i+1)r-s}, for the $j$-th ($s-r\leq j\leq s-1$) term we need to represent $\frac{(\hat{l}+n-ki)^{r'}}{r'n-ki}$ as a linear combination of $\hat{l\vphantom{}}^{\underline{(i+1)r-j+1-t}}/\hat{l\vphantom{}}^{\underline{ir-j}},\;0\leq t\leq r+1$. And the contribution is given by the coefficient of $\hat{l\vphantom{}}^{\underline{(i+1)r-s}}/\hat{l\vphantom{}}^{\underline{ir-j}}$. So the Lemma follows from the following Lemma \ref{lem:claim} if $s\leq ir$.

Next we consider $s>ir$. For $s>(i+1)r$, equation \eqref{eq: induction for g} is true since both sides are equal to 0. For $ir< s\leq (i+1)r$, the argument is the same as that in the case $s\leq ir$, but the contribution \eqref{form: total contribution to l (i+1)r-s} will only contain 
\begin{equation}\label{form: total contribution of s>ir}
    \frac{(\hat{l}+n-ki)^{r'}}{r'n-ki}\sum_{j=s-r}^{ir}g_j(i,n)\hat{l\vphantom{}}^{\underline{ir-j}}.
\end{equation}
Thus equation \eqref{eq: induction for g} is still true by the same argument for the terms in \eqref{form: total contribution of s>ir}. The lemma is thus completely proved.
\end{proof}

\begin{lemma}\label{lem:claim}
    The coefficient of $\hat{l\vphantom{}}^{\underline{(i+1)r-j+1-t}}/\hat{l\vphantom{}}^{\underline{ir-j}},\;0\leq t\leq r+1$ in $(\hat{l}+n-ki)^{r'}$ is a polynomial in $n$ which depends on $k,r,j,t$. In particular, the coefficients are independent of $i$.
\end{lemma}

\begin{proof}
    As in the case of $j=s$, 
\begin{align*}
    &\hat{l\vphantom{}}^{\underline{(i+1)r-j+1}}/\hat{l\vphantom{}}^{\underline{ir-j}}=\big(\hat{l}-k(i-1-[j]_r)\big)^{\langle j\rangle_r}\big(\hat{l}-k(i-[j]_r)\big)^{r-\langle j\rangle_r}\big(\hat{l}-k(i-[j]_r+\delta_{\langle j\rangle_r})\big).
\end{align*}
So the polynomials $\hat{l\vphantom{}}^{\underline{(i+1)r-j+1-t}}/\hat{l\vphantom{}}^{\underline{ir-j}},\;0\leq t\leq r+1$ are successively discarding terms in this expression from right to left. We first expand $(\hat{l}+n-ki)^{r'}$ into polynomials of $\hat{l}-k(i-1-[j]_r)$, which is
\begin{equation}\label{eq: first change}
   (\hat{l}+n-ki)^{r'}=\sum_{t=0}^{r'}\binom{r'}{t}(n-k-k[j]_r)^t(\hat{l}-k(i-1-[j]_r))^{r'-t}. 
\end{equation}

When further changing the basis $(\hat{l}-k(i-1-[j]_r))^a$ to 
\begin{equation}\label{eq:furtherchange}
    (\hat{l}-k(i-1-[j]_r))^{\langle j\rangle_r}\cdot (\hat{l}-k(i-[j]_r))^{r'-t-\langle j\rangle_r},
\end{equation}
one needs to expand $(\hat{l}-k(i-1-[j]_r))^a$ for $a\geq\langle j\rangle_r$ as
\begin{align}
    &(\hat{l}-k(i-1-[j]_r))^{\langle j\rangle_r}\cdot(\hat{l}-k(i-[j]_r)+k)^{a-\langle j\rangle_r}\nonumber\\
    =&(\hat{l}-k(i-1-[j]_r))^{\langle j\rangle_r}\sum_{b=0}^{a-\langle j\rangle_r}\binom{a-\langle j\rangle_r}{b}k^b(\hat{l}-k(i-[j]_r))^{a-\langle j\rangle_r-b}.\label{eq: second change}
\end{align}
In the above equation,
the term satisfying $r'-t=a-b$ contributes to \eqref{eq:furtherchange}. Combining Equation \eqref{eq: first change} and \eqref{eq: second change}, we see that for $r'-t\geq \langle j\rangle_r$, the coefficient of 
\[(\hat{l}-k(i-1-[j]_r))^{\langle j\rangle_r}\cdot (\hat{l}-k(i-[j]_r))^{r'-t-\langle j\rangle_r}\]
in $(\hat{l}+n-ki)^{r'}$ is
\begin{equation}\label{form: to compute degree of n}
    \sum_{a=\langle j\rangle_r}^{r'}\binom{r'}{a}(n-k-k[j]_r)^{r'-a}\binom{a-\langle j\rangle_r}{a+t-r'}k^{a+t-r'}.
\end{equation}
One sees that the coefficients are polynomials in $n$ which depends on $k,r,j,t$ if $r'-t\geq \langle j\rangle_r$. For $r'-t< \langle j\rangle_r$ this conclusion follows directly from equation \eqref{eq: first change} because the change of basis \eqref{eq: second change} does not affect them.

If $\langle j\rangle_r\neq0$, then
\begin{align*}
    \hat{l\vphantom{}}^{\underline{(i+1)r-j+1}}/\hat{l\vphantom{}}^{\underline{ir-j}}=(\hat{l}-k(i-1-[j]_r))^{\langle j\rangle_r}\cdot (\hat{l}-k(i-[j]_r))^{r'-\langle j\rangle_r},
\end{align*}
and the claim in this case is proved. If $\langle j\rangle_r=0$, we need to write one $\hat{l}-k(i-[j]_r)$ as
\[(\hat{l}-k(i-[j]_r+1))+k.\]
Hence 
\begin{align*}
   &(\hat{l}-k(i-1-[j]_r))^{\langle j\rangle_r}\cdot (\hat{l}-k(i-[j]_r))^{r'-\langle j\rangle_r} \\
   =&\hat{l\vphantom{}}^{\underline{(i+1)r-j+1}}/\hat{l\vphantom{}}^{\underline{ir-j}}+k\cdot\hat{l\vphantom{}}^{\underline{(i+1)r-j}}/\hat{l\vphantom{}}^{\underline{ir-j}}.
\end{align*}
So in this case, the coefficient of $\hat{l\vphantom{}}^{\underline{(i+1)r-j}}/\hat{l\vphantom{}}^{\underline{ir-j}}$ is the sum of two polynomials, and other coefficient is a single polynomial, which does not affect the conclusion of the claim. We have thus completed the proof of this lemma.
\end{proof}

Using Lemma \ref{lem: induction for g}, we can obtain the following

\begin{theorem}\label{thm: polynomiality of g}
    The $g_s(i,n)$ is a polynomial in $i\in\mathbb{Z}_{\geq0}$, with degree not bigger than $s$. And it is a rational function in $n$, with possible simple poles at 
    
    \[n=\frac{kt}{r'},\;0< t<s.\]
\end{theorem}

\begin{proof}
    Since $g_0(i,n)=1$, this is true for $s=0$. Now for $s>0$, suppose the theorem is true for less than $s$ case. Dividing both sides of equation \eqref{eq: induction for g} by $(r'n-ki-k)_s$, we get
    \begin{align*}
        \frac{g_s(i+1,n)}{(r'n-ki-k)_s}=\frac{g_s(i,n)}{(r'n-ki)_s}+\frac{\sum_{j=s-r}^{s-1}C_{s,j}(n)g_j(i,n)}{(r'n-ki)_{s+1}}.
    \end{align*}
    Recall that $(\cdots)_s$ is the $k$-Pochhammer symbol defined in Definition \ref{def: Pochhammer}. Expand 
    \[\sum_{j=s-r}^{s-1}C_{s,j}(n)g_j(i,n)=\sum_{t=0}^{s-1}a_t(r'n-ki)_t,\]
    with respect to $i$, where $a_t$ does not depend on $i$, and only has possible poles $n=\frac{ku}{r'}, 0<u<s-1$. Thus we have
    \begin{align*}
        \frac{g_s(i+1,n)}{(r'n-ki-k)_s}=\frac{g_s(i,n)}{(r'n-ki)_s}+\sum_{t=0}^{s-1}\frac{a_t}{(r'n-ki-kt)_{s+1-t}}.
    \end{align*}
    Consequently, we have
    \begin{align*}
        \frac{g_s(i,n)}{(r'n-ki)_s}=\frac{g_s(0,n)}{(r'n)_s}+\sum_{u=0}^{i-1}\sum_{t=0}^{s-1}\frac{a_t}{(r'n-ku-kt)_{s+1-t}}.
    \end{align*}
    The hypothesis $s>0$ implies $g_s(0,n)=0$, so the right hand side of the above equation is equal to
    \begin{align*}
        &\sum_{t=0}^{s-1}\sum_{u=0}^{i-1}\frac{a_t}{k(s-t)}\bigg(\frac{1}{(r'n-ku-kt-k)_{s-t}}-\frac{1}{(r'n-ku-kt)_{s-t}}\bigg)\\
        =&\sum_{t=0}^{s-1}\frac{a_t}{k(s-t)}\bigg(\frac{1}{(r'n-ki-kt)_{s-t}}-\frac{1}{(r'n-kt)_{s-t}}\bigg).
    \end{align*}
    Thus
    \begin{align*}
        g_s(i,n)=&\sum_{t=0}^{s-1}\frac{a_t}{k(s-t)}(r'n-ki)_t-(r'n-ki)_s\cdot\sum_{t=0}^{s-1}\frac{a_t}{k(s-t)(r'n-kt)_{s-t}},
    \end{align*}
and it is a polynomial in $i$ of degree $\leq s$, with possible poles $n=\frac{kt}{r'}, 0\leq t<s$. Note that $a_t$ may have poles in $n$. To prove the theorem, we only need to show that $g_s(i,n)$ does not have pole at $n=0$, and all its poles are simple.

    \vspace{10pt}

    \begin{claim}
        Given polynomials $f(i,n)$ and $g(n)$ in $\mathbb{C}[i,n]$, if for infinitely many $i\in \mathbb{C}$, $f(i,n)$ is divisible by $g(n)$ in $\mathbb{C}[n]$, then $f(i,n)$ is divisible by $g(n)$ in $\mathbb{C}[i,n]$.
    \end{claim}

    \vspace{5pt}

    \noindent \emph{Proof of claim}. We can write $f(i,n)=g(n)q(i,n)+r(i,n)$ where $r(i,n)$ has degree in $n$ smaller than the degree of $g(n)$. Then for infinitely many $i$, $r(i,n)=0$, hence $r(i,n)$ itself vanishes. And we have thus proved the claim.

    \vspace{10pt}
    
    We now return to the recurrence \eqref{eq: recurrence for G}. One can use induction to show that $G_n^{(i)}(l)$ has at most simple poles at $n=\frac{kt}{r'}, 0< t<i$. And since $G_n^{(1)}(l)=\frac{(\hat{l}+n)^{r'}}{r'n}-\frac{\hat{l}^{r'}}{r'n}$ does not have pole at $n=0$, by recurrence \eqref{eq: recurrence for G}, $G_n^{(i)}(l)$ still does not have pole at $n=0$. Suppose the order of pole $n=\frac{kt}{r'},\;0\leq t<s$ of $g_s(i,n)$ is $e_t\geq0$. We will show that $e_t\leq 1$ for any $0< t<s$ and $e_0=0$. Let 
    \[f(i,n)=\prod_{t=0}^{s-1}(r'n-kt)^{e_t}g_s(i,n)\in \mathbb{C}[i,n].\]
    Since $g_s(i,n)$ is the coefficient of $G_n^{(i)}(l)$, for any specific $i\in\mathbb{Z}_{\geq0}$ it has at most simple poles at $n=\frac{kt}{r'},\;t>0$. Hence 
    \begin{equation}\label{eq: remaining factor}
        n^{e_0}\prod_{t=1}^{s-1}(r'n-kt)^{e_t-1}
    \end{equation}
    divides $f(i,n)$ for any $i\in\mathbb{Z}_{\geq0}$. Using the above claim, the factor \eqref{eq: remaining factor} divides $f(i,n)$ in $\mathbb{C}[i,n]$. However, if $e_0>0$ or $e_t>1$ for some $0<t<s$, this will contradict with the definition of the order of the poles.
\end{proof}

\begin{corollary}\label{cor: polynomiality wrt old polynomial basis}
    When $G_n^{(i)}(l)$ is expanded with respect to the basis $\{l^{ir-s}\;|\;s\geq 0\}$, the coefficient of $l^{ir-s}$ for a fixed $s$ is a polynomial in $i$, and rational in $n$ with simple poles described in Theorem \ref{thm: polynomiality of g}.
\end{corollary}

\begin{proof}
    This follows from Lemma \ref{lem: base change between polynomial basis} and Theorem \ref{thm: polynomiality of g}.
\end{proof}

\section{Quasi-polynomiality of single connected $k$-leaky $(r+1)$-completed Hurwitz numbers}\label{sec:quasipolynomiality}

This section is devoted to the proof of quasi-polynomiality of single connected leaky Hurwitz numbers using the results in Section \ref{sec: bosons}. We will let $g_s(i,n)$ be the coefficient of $G_n^{(i)}(l)$ with respect to the basis $\{l^{ir-s}\;|\;s\geq 0\}$ by the analogous expression of equation \eqref{eq: expansion of G} throughout this section. By Corollary \ref{cor: polynomiality wrt old polynomial basis}, they are polynomials in $i$ and rational in $n$ with possible simple poles at $n=\frac{kt}{r'},\;0< t<s$.

\subsection{The $\mathcal{A}$-operators}

Let $\mu=(\mu_1,\mu_2,\cdots,\mu_N),\;\sigma=2g-2+2N$ and $\tau=2g-2+N$. By letting $i=[n]+a$ and $b=t+ka-r[n]$ and extracting the common factor
\begin{equation}\label{eq: def of d}
    d(n):=\frac{\hbar^{[n]}(r'n-k\sigma)_{[n]-\sigma}}{[n]!(r'!)^{[n]}},
\end{equation}
from the first part on the right hand side of equation \eqref{eq: second conjugate}, we define $\mathcal{A}$-operators as
\begin{align}
    \mathcal{A}_{\langle n\rangle}(\hbar,n):=\sum_{a\in\mathbb{Z}}\frac{\hbar^a(r'n)_\sigma(r'n-k[n])_a}{n(r'!)^a([n]+1)_{\overline{a}}}\sum_{t\geq0}\sum_{l\in\mathbb{Z}+\frac{1}{2}}\frac{\nabla^t}{t!}G_n^{([n]+a)}(l)E_{l+\langle n\rangle-b,l}.\label{form: def A operator}
\end{align}
This operator is motivated by previous works, see for example \cite{DKOSS,DLPS,KLPS,OP06}. Here $a$ runs through all integers since $\frac{1}{([n]+1)_{\overline{a}}}=0$ for $a< -[n]$. So from equation \eqref{eq: hurwitz as prod of A operator} and \eqref{eq: second conjugate} we have
\begin{equation}\label{eq: generating func as A operators}
    h^{k,r}_\mu(\hbar)=\frac{1}{|\mathrm{Aut}(\mu)|}\prod_{j=1}^{N}d(\mu_j)\langle\prod_{j=1}^N(\mathcal{A}_{\langle \mu_j\rangle}(\hbar,\mu_j)+C_j)\rangle.
\end{equation}
Here $C_j$ represents some scalar.

We first recall

\begin{lemma}\label{lem: difference polynomial}(Lemma 4.5 in \cite{KLPS})
    The coefficient of $l^q$ in 
    
    \[\frac{\nabla^{x+m}}{(x+m)!}l^{p+x}\]
    is polynomial in $x$ for fixed $p$ and $m$, with degree $2p-q-2m$.
\end{lemma}

Note that we factor out $\hbar^{\sum_j[\mu_j]}$ in $\prod_{j=1}^{N}d(\mu_j)$. So the remaining power of $\hbar$ determined by equation \eqref{eq: relation of g and s} is $\hbar^{\frac{\tau+\sum_j\langle \mu_j\rangle}{k+r}}$. Since expression \eqref{eq: sym in nu} is symmetric in their arguments, we will first consider the dependence on $[\mu_1]$.

\begin{proposition}\label{prop: disc corre of A operator}
    For fixed $\langle \mu_1\rangle,\cdots,\langle \mu_N\rangle$ and $[\mu_2],\cdots,[\mu_N]$, the coefficient $\hbar^{\frac{\tau+\sum_j\langle \mu_j\rangle}{k+r}}$ in
    \begin{equation}\label{form: prod of A operators}
        \langle\prod_{j=1}^N\mathcal{A}_{\langle \mu_j\rangle}(\hbar,\mu_j)\rangle,
    \end{equation}
    is a rational function in $[\mu_1]$, with possible simple poles at negative integers and $-\frac{\langle \mu_1\rangle}{k+r}$. And it has an upper bound of degree in $[\mu_1]$ which does not depend on $[\mu_2],\cdots,[\mu_N]$.
\end{proposition}

\begin{proof}
    We collect the constraints on parameters $a_j$ and $b_j,\;1\leq j\leq N$.
    \begin{enumerate}
    \item $a_j\geq -[\mu_j]$ and $b_j\geq ka_j-r[\mu_j]$ for any $1\leq j\leq N$.
        \item The energy $b_1-\langle \mu_1\rangle$ of $\mathcal{A}_{\langle \mu_1\rangle}(\hbar,\mu_1)$ should be non-negative, which means $b_1\geq \langle \mu_1\rangle$.
        \item The exponent $r[\mu_j]+b_j-ka_j$ of the backward difference operator in formula \eqref{form: def A operator} is at most the degree $r[\mu_j]+ra_j$ of $G_{\mu_j}^{[\mu_j]+a_j}(l)$ in $l$, which is $b_j\leq (k+r)a_j$.
        \item The total energy of operators in expression \eqref{form: prod of A operators} should be 0, thus 
        
        \[\sum_jb_j=\sum_j\langle \mu_j\rangle.\]
        \item The power of $\hbar$ is $\sum_j a_j$ which is equal to $\frac{\tau+\sum_j\langle \mu_j\rangle}{k+r}$.
    \end{enumerate}

    By point 2 and point 3, we know that $a_1$ and $b_1$ have lower bound. In particular, $a_1\geq0$, so the term $((r+1)\mu_1-k[\mu_1])_{a_1}$ of formula \eqref{form: def A operator} does not give poles. Since $[\mu_j]$ are fixed, by point 1, $a_j$ has lower bound for $j\geq2$. Since the summation $\sum_{j=1}^Na_j$ is fixed by point 5, $a_1$ has upper bound. Then by point 3, $b_1$ also has upper bound. So the summation indices $a$ and $t$ of formula \eqref{form: def A operator} for $n=\mu_1$ are finite (possibly depend on $[\mu_j],\;j\geq2$), and the summands are rational functions in $[\mu_1]$ (note that fixed $a$ and $t$, $\langle0|E_{l+\langle n\rangle-b,l}\neq 0$ only for finite $l$). There are possible poles of $\mu_1=0$, $[\mu_1]$ at negative integers and at $\mu_1=\frac{kp}{r'},\;p>0$ from Theorem \ref{thm: polynomiality of g}.
    
    Since $\nabla^{r[\mu_1]+b_1-ka_1}$ annihilates terms with powers in $l$ smaller than $r[\mu_1]+b_1-ka_1$, we know that $\frac{\nabla^{r[\mu_1]+b_1-ka_1}}{(r[\mu_1]+b_1-ka_1)!}G_{\mu_1}^{([\mu_1]+a_1)}(l)$ in formula \eqref{form: def A operator} contains only $g_s([\mu_1]+a_1,\mu_1)$ with
    \begin{align}\label{eq: inequ tau}
    \begin{split}
        s&\leq r([\mu_1]+a_1)-(r[\mu_1]+b_1-ka_1)\\
        &=(k+r)a_1-b_1\\
        &=(k+r)(\frac{\tau+\sum_j\langle \mu_j\rangle}{k+r}-\sum_{j=2}^Na_j)-\sum_j\langle \mu_j\rangle+\sum_{j=2}^Nb_j\\
        &=\tau+\sum_{j=2}^N(b_j-(k+r)a_j)\leq \tau.
        \end{split}
    \end{align}
    So there is an upper bound of the power in $[\mu_1]$ coming from $g_s([\mu_1]+a_1,\mu_1)$ which does not depend on $[\mu_2],\cdots,[\mu_N]$. By Lemma \ref{lem: difference polynomial} and the same argument as in inequality \eqref{eq: inequ tau}, the contribution of   
    \[\frac{\nabla^{r[\mu_1]+b_1-ka_1}}{(r[\mu_1]+b_1-ka_1)!}l^{r[\mu_1]+ra_1-s}\]
    has degree in $[\mu_1]$ not bigger than
    \[2(ra_1-s)-2(b_1-ka_1)\leq 2\tau,\]
    which does not depend on $[\mu_2],\cdots,[\mu_N]$. The degree of $[\mu_1]$ in    
    \[\frac{(r'\mu_1)_\sigma(r'\mu_1-k[\mu_1])_a}{\mu_1([\mu_1]+1)_{\overline{a}}}\]
    is $\sigma-1$, and the total power of $r'!$ is $\frac{\tau+\sum_j\langle \mu_j\rangle}{k+r}$. The simple poles coming from $g_s([\mu_1]+a_1,\mu_1)$ can occur at $\mu_1=\frac{kp}{r'},\;0<p<\tau$, which cancel out with $(r'\mu_1)_\sigma$. We have thus proved the proposition.

\end{proof}

Note that the generating function of single connected Hurwitz numbers is related to the disconnected one via
\begin{equation}\label{eq: discon via con}
    |\mathrm{Aut}(\mu)|h_\mu^{k,r}(\hbar)=\sum_{M\vdash\{1,\cdots,N\}}\prod_{J\in M}|\mathrm{Aut}(\mu_J)|h_{\mu_J}^{k,r,\circ}(\hbar),
\end{equation}
where $M$ runs through all partitions of $\{1,\cdots,N\}$, and $\mu_J$ is the sub-partition of $\mu$ corresponding to subset $J$. Elements $J$ of $M$ are non-empty subsets of $\{1,\cdots,N\}$. Hence $h_\mu^{k,r,\circ}(\hbar)$ is obtained by inverting equation \eqref{eq: discon via con}, and we have
\begin{equation}\label{eq: inclusion-exclusion}
    |\mathrm{Aut}(\mu)|h_\mu^{k,r,\circ}(\hbar)=\sum_{M\vdash\{1,\cdots,N\}}(-1)^{|M|-1}(|M|-1)!\prod_{J\in M}|\mathrm{Aut}(\mu_J)|h_{\mu_J}^{k,r}(\hbar).
\end{equation}

\begin{lemma}
    If $N>1$, we have
\begin{equation}\label{eq: independ on C}
    h_\mu^{k,r,\circ}(\hbar)=\frac{1}{|\mathrm{Aut}(\mu)|}\prod_{j=1}^{N}d(\mu_j)\langle\prod_{j=1}^N\mathcal{A}_{\langle \mu_j\rangle}(\hbar,\mu_j)\rangle^\circ,
\end{equation}
where $d(\mu_j)$ is defined in equation \eqref{eq: def of d}.
\end{lemma}

\begin{proof}

 Recall equation \eqref{eq: generating func as A operators}, then we will prove equation \eqref{eq: independ on C} by showing that
\begin{equation}\label{eq: conn correlator of constant}
    \langle C_{j_0}\prod_{j\neq j_0}\mathcal{W}_j\rangle^\circ=0,
\end{equation}
for any $1\leq j_0\leq N$, where $\mathcal{W}_j$ are some operators. We divide partitions of $\{1,\cdots,N\}$ into two classes, the first contains the subset $\{j_0\}$ and the latter does not. For any $M'=\{J_1,\cdots, J_p\}\vdash\{1,\cdots,\hat{j}_0,\cdots,N\}$ where $\hat{j}_0$ means omitting $j_0$, let 
\[M'_u:=\{J_1,\cdots,J_u\cup\{j_0\},\cdots,J_p\}\vdash\{1,\cdots,N\},\]
for $1\leq u\leq p$. Then
\[(-1)^{p}p!\langle C_{j_0}\rangle\prod_{J\in M'}\langle \prod_{j\in J}\mathcal{W}_j\rangle\]
will cancel out with
\[\sum_{u=1}^p(-1)^{p-1}(p-1)!\langle C_{j_0}\prod_{j\in J_u}\mathcal{W}_j\rangle\prod_{\substack{J\in M'\\J\neq J_u}}\langle \prod_{j\in J}\mathcal{W}_j\rangle.\]
We have thus proved equation \eqref{eq: conn correlator of constant}, which implies equation \eqref{eq: independ on C}.

\end{proof}

\begin{remark}\label{rm: sigma > tau nuJ}
    Since connected leaky Hurwitz numbers count tropical covers of $\mathbf{P}^1$ with connected source curve, it has non-negative genus. Hence $\tau(\mu_J)=2g-2+l(\mu_J)\geq-1$, and the degree in $\hbar$ gives 
    \[\frac{\tau+|\mu|}{k+r}=\sum_{J\in M}\frac{\tau(\mu_J)+|\mu_J|}{k+r},\]
    hence $\tau=\sum_{J\in M}\tau(\mu_J)$. We know that $\tau(\mu_J)\leq \tau+N-1<\sigma$, so the same argument as in Proposition \ref{prop: disc corre of A operator} shows that the correlators in computing the connected one do not have poles at $\mu_1=\frac{kp}{r'},\;p>0$.
\end{remark}

\subsection{Proof of quasi-polynomiality}

We prove in this subsection the first main result of this paper.

\begin{theorem}\label{thm: quasi-polynomiality}
    The single connected leaky Hurwitz numbers $h_{g,\mu}^{k,r,\circ}$ are quasi-polynomials in $\mu=(\mu_1,\cdots,\mu_N)$ in the stable $\tau=2g-2+N>0$ case. Let $s=\frac{\tau+|\mu|}{k+r}$, we have:
    \begin{itemize}
        \item If $N>1$, then
        \[h_{g,\mu}^{k,r,\circ}=\frac{s!}{|\mathrm{Aut}(\mu)|}\prod_{i=1}^N\frac{(r'\mu_i-k\sigma)_{[\mu_i]-\sigma}}{[\mu_i]!(r'!)^{[\mu_i]}}P_{\langle \mu_1\rangle,\cdots,\langle \mu_N\rangle}([\mu_1],\cdots,[\mu_N]),\]
        where $P_{\langle \mu_1\rangle,\cdots,\langle \mu_N\rangle}$ is a polynomial depending on parameters $\langle \mu_1\rangle,\cdots,\langle \mu_N\rangle$.
        \item If $N=1$, $h^{k,r,\circ}_{g,(n)}=h^{k,r}_{g,(n)} $ is not 0 only if $n\equiv1-2g\;(\mathrm{mod}\;(k+r))$ and $r[n]_k\geq2g-1+\langle n\rangle_k$. For $r[n]_k>2g-1+\langle n\rangle_k$, $h^{k,r}_{g,(n)}$ is of the form
        \[\frac{(r'n-2kg)_{s-2g}}{(n-ks)(r'!)^s}P(n),\]
        where $P(n)$ is a polynomial in $n$.
    \end{itemize}
\end{theorem}

We first compute residues for $\mathcal{A}_{\langle n\rangle}(\hbar,n)$ of $[n]$ at negative integers. Let $n_-:=\langle n\rangle-(k+r)m<0$ for a positive integer $m$. By the same argument as Proposition \ref{prop: second conjugation}, and notice that in this case the initial positive energy leads to no scalar multiplication term, we get
\begin{align*}
    e^{\alpha_1}e^{\hbar\mathcal{H}_{k,r}}\alpha_{-n_-}e^{-\hbar\mathcal{H}_{k,r}}e^{-\alpha_1}
    =&\sum_{i=0}^{\infty}\sum_{t=0}^{\infty}\frac{\hbar^i}{i!t!}\sum_{l\in \mathbb{Z}+\frac{1}{2}}\nabla^t F_{n_-}^{(i)}(l)E_{l+n_--ki-t,l},\\
    =&\sum_{a\geq m}\frac{\hbar^{a-m}(r'n_-)_{a-m}}{(a-m)!(r'!)^{a-m}}\sum_{t\geq0}\sum_{l\in \mathbb{Z}+\frac{1}{2}}\frac{\nabla^t}{t!}G_{n_-}^{(a-m)}(l)E_{l+n_--ka+km-t,l}.
\end{align*}
By formula \eqref{form: def A operator}, the residue of $\mathcal{A}_{\langle n\rangle}(\hbar,n)$ at negative integer $-m$ is
\begin{align*}
    &\text{Res}_{[n]=-m}\mathcal{A}_{\langle n\rangle}(\hbar,n)\\
    =&\sum_{a\in\mathbb{Z}}\frac{\hbar^a(r'n)_\sigma(r'n-k[n])_a([n]+m)}{n(r'!)^a([n]+1)_{\overline{a}}}\sum_{t\geq0}\sum_{l\in\mathbb{Z}+\frac{1}{2}}\frac{\nabla^t}{t!}G_n^{([n]+a)}(l)E_{l+n-ka-k[n]-t,l}\bigg|_{[n]=-m}\\
    =&\sum_{a\geq m}(-1)^{m-1}\frac{\hbar^a(r'n_-)_\sigma(r'n_-+km)_a}{n_-(m-1)!(a-m)!(r'!)^a}\sum_{t\geq0}\sum_{l\in\mathbb{Z}+\frac{1}{2}}\frac{\nabla^t}{t!}G_{n_-}^{(a-m)}(l)E_{l+n_--ka+km-t,l}\\
    =&(-1)^{m-1}\frac{\hbar^m(r'n_-)_\sigma(r'n_-+km)_m}{n_-(m-1)!(r'!)^m}e^{\alpha_1}e^{\hbar\mathcal{H}_{k,r}}\alpha_{-n_-}e^{-\hbar\mathcal{H}_{k,r}}e^{-\alpha_1},
\end{align*}
where in the second equality the fact that denominator has factor $([n]+m)$ requires that $a\geq m$. Define
\[h(\langle n\rangle,m):=(-1)^{m-1}\frac{\hbar^m(r'n_-)_\sigma(r'n_-+km)_m}{n_-(m-1)!(r'!)^m}.\]
We are now ready to prove

\begin{proposition}\label{prop: conn corre of A}
    If $N>1$ and $\tau=2g-2+N>0$, the coefficient of $\hbar^{\frac{\tau+\sum_j\langle \mu_j\rangle}{k+r}}$ in the connected correlator
    \[\langle\prod_{j=1}^N\mathcal{A}_{\langle \mu_j\rangle}(\hbar,\mu_j)\rangle^\circ,\]
    is polynomial in $[\mu_1],\cdots,[\mu_N]$ for fixed $\langle \mu_1\rangle,\cdots,\langle \mu_N\rangle$.
\end{proposition}

\begin{proof}
    We first consider the dependence on $\mu_1$. We show that the residues of Proposition \ref{prop: disc corre of A operator} vanish. The residue at $[\mu_1]=-m$ is
    \begin{align*}
        &\text{Res}_{[\mu_1]=-m}\sum_{M\vdash\{1,\cdots,N\}}(-1)^{|M|-1}(|M|-1)!\prod_{J\in M}\langle\prod_{j\in J}\mathcal{A}_{\langle \mu_j\rangle}(\hbar,\mu_j)\rangle\\
        =&\sum_{M\vdash\{1,\cdots,N\}}(-1)^{|M|-1}(|M|-1)!\frac{h(\langle \mu_1\rangle,m)}{\prod_{j\neq 1}\mu_jd(\mu_j)}\prod_{J\in M}\langle e^{\alpha_1}e^{\hbar\mathcal{H}_{k,r}}{\prod_{j\in J}}^-\alpha_{-\mu_{j}}\rangle,
    \end{align*}
    where $\prod^-$ means that when $j=1$, one should multiply $\alpha_{-n_-}$ instead of $\alpha_{-\mu_1}$. Here we have used equation \eqref{eq: conn correlator of constant}. We divide partitions of $\{1,\cdots,N\}$ into three classes: The first contains the subset $\{1\}$; The second contains a two-element subset $\{1,j\}$; And the third are the remaining ones.
    
    For a partition $M$ in the first class, since $n_-<0$, we have
    \[\langle e^{\alpha_1}e^{\hbar\mathcal{H}_{k,r}}\alpha_{-n_-}\rangle=0.\]
    So the residue of these components at $[\mu_1]=-m$ is 0. Let $M=\{J_1,\cdots,J_p\}$ be a partition in the third class, suppose $1\in J_u$ where $1\leq u\leq p$. Then $J_u$ contains at least three elements. We compute  
    \[\langle e^{\alpha_1}e^{\hbar\mathcal{H}_{k,r}}{\prod_{j\in J_u}}^-\alpha_{-\mu_{j}}\rangle\]
    by commuting $\alpha_{-n_-}$ to the rightmost. Since $\alpha_{-n_-}$ annihilates $|0\rangle$, the correlator equals
    \[\sum_{\substack{j\in J_u\\j\neq1}}
    (-n_-) \delta_{n_-+\mu_j}
    \langle e^{\alpha_1}e^{\hbar\mathcal{H}_{k,r}} \prod_{\substack{t\in J_u\\t\neq 1,j}}
    \alpha_{-\mu_{t}} \rangle,\]
    which is equal to
    \[\sum_{\substack{j\in J_u\\j\neq1}}
    \langle e^{\alpha_1}e^{\hbar\mathcal{H}_{k,r}} \alpha_{n_-} \alpha_{-\mu_j}\rangle
    \langle e^{\alpha_1} e^{\hbar\mathcal{H}_{k,r}} \prod_{\substack{t\in J_u\\t\neq 1,j}}\alpha_{-\mu_{t}}\rangle.\]
    Let $J_v'=J_v$ if $v\neq u$ and $J_u'=J_u\backslash\{1,j\}$.
    This correspondence shows that the residue for partition
    \[M'=\{\{1,j\}\}\cup\{J_1',\cdots,J_p'\}\]
    in the second class appears $p$ times in the residue for partitions 
    \[M''=\{J_1',\cdots,J_v'\cup\{1,j\},\cdots,J_p'\}.\]
    The factor in the summand of $M'$ differs from that of $M''$ by $-p$,
    so they cancel out with each other.
    If $N=2$,
    the residue at $[\mu_1]=-m$,
    \begin{align}
        \frac{h(\langle \mu_1\rangle,m)}{\mu_2d(\mu_2)}
        \langle e^{\alpha_1} e^{\hbar\mathcal{H}_{k,r}}\alpha_{-n_-}\alpha_{-\mu_2}\rangle
        =(-n_-)\frac{h(\langle \mu_1\rangle,m)}{\mu_2d(\mu_2)}\delta_{n_-+\mu_2},
        \label{eq: N=2}
    \end{align}
    does not cancel out with other terms.
    It is non-zero only if
    $\mu_2=m(k+r)-\langle \mu_1\rangle$. In this case
    $[\mu_2]=m-1+\delta_{\langle \mu_1\rangle}$
    and $\langle \mu_2\rangle=(k+r)(1-\delta_{\langle \mu_1\rangle})-\langle \mu_1\rangle$.
    We have
    \[\frac{\tau+\langle \mu_1\rangle+\langle \mu_2\rangle}{k+r}=\frac{\tau}{k+r}+1-\delta_{\langle \mu_1\rangle}>0.\]
    But formula \eqref{eq: N=2} has degree $m-[\mu_2]=1-\delta_{\langle \mu_1\rangle}$ in $\hbar$. Since $\tau>0$, it does not contribute to $\hbar^{\frac{\tau+\sum_j\langle \mu_j\rangle}{k+r}}$.
    So the connected correlator has no poles at negative integers in any case.

    The pole at $[\mu_1]=-\frac{\langle \mu_1\rangle}{k+r}$ can occur only if $\sigma=0$ by equation \eqref{form: def A operator}. This is possible only if $g=0, \,N=1$, so $[\mu_1]$ always does not have poles. Since connected correlators are symmetric in their inputs by equation \eqref{eq: inclusion-exclusion}, and the degree in $[\mu_1]$ does not depend on $[\mu_2],\cdots,[\mu_N]$ by Proposition \ref{prop: disc corre of A operator}, we complete the proof.
\end{proof}

This proposition, together with equation \eqref{eq: independ on C} gives the result of the first part of Theorem \ref{thm: quasi-polynomiality}.

For $N=1$, since $E_{l,l}$ annihilates $|0\rangle$ for any $l\in\mathbb{Z}+\frac{1}{2}$, we have 
\begin{align*}
    h_n^\circ(\hbar)=h_n(\hbar)=C+C',
\end{align*}
where
\begin{align*}
    C=&\frac{\hbar^{[n]_k}\delta_{\langle n\rangle_k}}{[n]_k!r'!n}\sum_{\substack{l\in \mathbb{Z}+\frac{1}{2}\\
        0<l<k}}\hat{l}^{r'}F_{n}^{([n]_k-1)}(l-k)\\
    &+(1-\delta_{\langle n\rangle_k})\frac{\hbar^{[n]_k}}{[n]_k!\langle n\rangle_k!n}\nabla^{\langle n\rangle_k-1} F_n^{([n]_k)}(-\frac{1}{2}),\\
    C'=&\sum_{i=0}^{[n]_k-1}\frac{\hbar^i}{i!(n-ki)!n}\nabla^{n-ki-1} F_n^{(i)}(-\frac{1}{2}).
\end{align*}
Note that the coefficient in the required power of $\hbar$ in $C+C'$ is non-zero only if $s=\frac{2g-1+n}{k+r}\in \mathbb{Z}\cap[0,[n]_k]$, which requires that $ n\equiv1-2g\;(\mathrm{mod}\;(k+r))$ and $r[n]_k\geq2g-1+\langle n\rangle_k$.

In the $r[n]_k>2g-1+\langle n\rangle_k$ case, the coefficient of $C$ is zero. The coefficient of $C'$ is
\begin{align}
    &\frac{(r'n)_{s}}{s!(n-ks)!n(r'!)^s}\nabla^{n-ks-1}G_n^{(s)}(-\frac{1}{2}).\label{eq: case 1}
\end{align}
Since $rs-(n-ks-1)=2g$ is fixed, as in the proof of Proposition \ref{prop: disc corre of A operator}, using Theorem \ref{thm: polynomiality of g} and Lemma \ref{lem: difference polynomial}, one gets that
\begin{equation}\label{eq: const rational part}
    \frac{\nabla^{n-ks-1}}{(n-ks-1)!}G^{(s)}_n(-\frac{1}{2})
\end{equation}
is a rational function in $n$, with possible simple poles at $n=\frac{kt}{r'},\,0<t<2g$. We can write 
\begin{equation}\label{eq:shift}
    (r'n)_{s}=(r'n-2kg)_{s-2g}(r'n)_{2g}.
\end{equation}
Since we are in the stable case, which requires that $g>0$, the $(r'n)_{2g}$ eliminates these poles and the $n$ in the denominator. We have thus completed the proof of Theorem \ref{thm: quasi-polynomiality}.

\begin{remark}\label{rem:criticalcase}
For fixed $g$, there are only finitely many $n$ satisfying $r[n]_k=2g-1+\langle n\rangle_k$. If $\langle n\rangle_k=0$, the coefficient of $\hbar^s$ only appears in $C$ and the Hurwitz number $h_{g,(n)}^{k,r}$ is
\begin{equation}\label{eq: case 2}
    \frac{s!(r'n)_{[n]_k-1}}{[n]_k!(r'!)^{[n]_k}n}\sum_{\substack{l\in \mathbb{Z}+\frac{1}{2}\\
        0<l<k}}\hat{l}^{r'}G_{n}^{([n]_k-1)}(l-k).
\end{equation}
If $\langle n\rangle_k\neq0$, it is
\begin{equation}\label{eq: case 3}
    \frac{s!(r'n)_{[n]_k}}{[n]_k!\langle n\rangle_k!n(r'!)^{[n]_k}}\nabla^{\langle n\rangle_k-1} G_n^{([n]_k)}(-\frac{1}{2}).
\end{equation}
\end{remark}

\section{Computations of single connected leaky completed Hurwitz numbers}\label{sec: conn n point function}
In this section,
we derive a closed formula to compute the single connected $k$-leaky $(r+1)$-completed Hurwitz numbers using KP integrability.
We also study applications of this formula.

\subsection{Connected $N$-point function for leaky Hurwitz numbers}
In this subsection, we provide a formula to compute the connected leaky Hurwitz numbers. The partition function of $k$-leaky $(r+1)$-completed possibly disconnected Hurwitz numbers is related to the connected one
\[H^{leaky}_{k,r}(\mathbf{t};\hbar):=\sum_{\mu\in\mathcal{P}}\sum_{s\geq 0}
\Big(\frac{1}{s!}\hbar^sh^{k,r,s,\circ}_\mu \prod_{i=1}^{l(\mu)}\mu_it_{\mu_i}\Big)\]
via
\begin{align}
    \tau_{k,r}^{leaky}(\mathbf{t};\hbar)=
    \exp \big(H^{leaky}_{k,r}(\mathbf{t};\hbar)\big)
    \in\mathbb{C}[\![\mathbf{t}, \hbar]\!].\label{eq:conn disconn}
\end{align}
So to compute the connected Hurwitz numbers,
we need to obtain the expansion formula of $\log \tau_{k,r}^{leaky}(\mathbf{t};\hbar)$. Formula \eqref{eq: partition function}, together with (see, for example, equation (A.16) in \cite{O01})
\begin{align*}
    \Gamma_-(\mathbf{t})|0\rangle
    =\sum_{\mu\in\mathcal{P}}
    s_{\mu}(\mathbf{t}) \cdot |\mu\rangle,
\end{align*}
provides the Schur polynomial expansion of the tau-function $\tau_{k,r}^{leaky}(\mathbf{t};\hbar)$.
Consequently,
we have
\begin{align}\label{eqn:leaky tau as schur}
    \tau_{k,r}^{leaky}(\mathbf{t};\hbar)
    =\sum_{\mu\in\mathcal{P}}
    \langle0|e^{\alpha_1} e^{\hbar \mathcal{H}_{k;r}} |\mu\rangle \cdot s_{\mu}(\mathbf{t}).
\end{align}
For convenience,
we denote the coefficient of $s_{\mu}(\mathbf{t})$ by 
$c^{leaky}_\mu=\langle0|e^{\alpha_1} e^{\hbar \mathcal{H}_{k;r}} |\mu\rangle$.
\begin{proposition}\label{prop:Aleaky formula}
    The fermionic two-point function of the partition function of $k$-leaky $(r+1)$-completed Hurwitz numbers defined in \eqref{eqn:def Azw} is equal to
    \begin{align}\label{eqn: Aleaky as tildePQ}
        A^{leaky}_{k;r}(z,w)
        =\sum_{a\in\mathbb{N}+\half}\tilde{\mathcal{P}}_{-a}(w)\,\mathcal{Q}_{-a}(z),
    \end{align}
    where for $a,b\in\mathbb{Z}+\half$, we denote
    \begin{align}
        \tilde{\mathcal{P}}_a(w)&=w^{-a-\frac12}\!\!\sum_{i,t\geq0 \atop ik+t\geq-a+\half}
        \frac{\hbar^{i}\,w^{-ik-t}}{i!\,t!\,(r'!)^{i}}
        (\hat{a}+t+ik)_i^{r'}, \label{eqn:def tildeP}\\
        \mathcal{Q}_b(z)&=z^{b-\frac12}\!\!\sum_{p,t\geq0}
        \frac{(-\hbar)^{p}(-1)^{t}\,z^{-pk-t}}{p!\,t!\,(r'!)^{p}}
        (\hat{b}-t)_p^{r'}, \label{eqn:def Q}
    \end{align}
    where $()_i$ and $()_p$ are $k$-Pochhammer symbols given in Definition \ref{def: Pochhammer}.
\end{proposition}
\begin{proof}
From the Schur polynomial expansion \eqref{eqn:leaky tau as schur} of the partition function $\tau_{k;r}^{leaky}(\mathbf{t};\hbar)$ of $k$-leaky $(r+1)$-completed Hurwitz numbers,
the corresponding fermionic two-point function is given by
\begin{align}\label{eqn:Aleaky=1}
    A_{k;r}^{leaky}(z,w)
    =&\sum_{n,m\in\mathbb{N}}
    \langle0|e^{\alpha_1} e^{\hbar \mathcal{H}_{k;r}}  \psi_{m+\half}\psi^*_{-n-\half}|0\rangle\cdot z^{-n-1} w^{-m-1},
\end{align}
where we have used equation \eqref{eqn:def Azw} and $|(m|n)\rangle=(-1)^n\psi_{m+\half}\psi^*_{-n-\half}|0\rangle$.
Then from $\psi^*_{a+\half}|0\rangle=0, \psi_{-b-\half}\psi^*_{-a-\half}|0\rangle=\delta_{a,b}|0\rangle, a,b\in\mathbb{N}$,
and $\langle0|e^{\alpha_1} e^{\hbar \mathcal{H}_{k;r}} |0\rangle=1$,
equation \eqref{eqn:Aleaky=1} can be rewritten as
\begin{align}\label{eqn:Aleaky as psipsi*}
    A_{k;r}^{leaky}(z,w)
    =&\langle0| e^{\alpha_1} e^{\hbar \mathcal{H}_{k;r}}  \psi(w)\psi^*(z) |0\rangle
    -\sum_{h\geq 0} z^{-1-h}w^h
\end{align}
where $\psi(w):=\sum_{m\in\mathbb{Z}}\psi_{m+\frac{1}{2}}w^{-m-1},\,\psi^*(z):=\sum_{n\in\mathbb{Z}}\psi^*_{-n-\frac{1}{2}}z^{-n-1}$. In the following Lemma \ref{lem:conj on psipsi*},
we have computed the conjugation of the operator $e^{\alpha_1} e^{\hbar \mathcal{H}_{k;r}}$ on the fermionic fields $\psi(w)$ and $\psi^*(z)$,
which,
together with $\langle0|\psi_a\psi^*_b|0\rangle=\delta_{a,b}\delta_{b<0}$ gives the following formula
\begin{align}\label{eqn: Aleaky as PQ}
        A^{leaky}_{k;r}(z,w)
        =\sum_{a\in\mathbb{N}+\half}\mathcal{P}_{-a}(w)\,\mathcal{Q}_{-a}(z)
    \;-\;\sum_{h\geq0}z^{-1-h}w^{h},
    \end{align}
where the function $\mathcal{P}_{a}(w)$ is defined in equation \eqref{eqn:def P}.
Note that from the definition \eqref{eqn:def Azw},
we have $A_{k;r}^{leaky}(z,w)\in z^{-1}w^{-1}\mathbb{C}[\![z^{-1},w^{-1}]\!]$.
If we split $\mathcal{P}_{a}(w)$ into its nonnegative $w$-degree part and negative $w$-degree part as
\begin{align*}
    \mathcal{P}_{a}(w)
    =\mathcal{P}^{\geq0}_{a}(w)
    +\tilde{\mathcal{P}}_{a}(w),
\end{align*}
where the negative $w$-degree part $\tilde{\mathcal{P}}_{a}(w)$ is exactly the function defined in equation \eqref{eqn:def tildeP}.
Then the contribution from the nonnegative $w$-degree part $\mathcal{P}^{\geq0}_{a}(w)$ must cancel with the last term $\sum_{h\geq0}z^{-1-h}w^{h}$.
Thus equation \eqref{eqn: Aleaky as PQ} reduces to the formula \eqref{eqn: Aleaky as tildePQ} for the fermionic two-point function $A^{leaky}_{k;r}(z,w)$.
\end{proof}

\begin{lemma}\label{lem:conj on psipsi*}
We have the following conjugation formulas
    \begin{align}
        e^{\alpha_1} e^{\hbar \mathcal{H}_{k;r}}  \psi(w) e^{-\hbar \mathcal{H}_{k;r}} e^{-\alpha_1}
        =\sum_{a\in\mathbb{Z}+\frac12}\psi_a\,\mathcal{P}_a(w),\\
        e^{\alpha_1} e^{\hbar \mathcal{H}_{k;r}} \psi^*(z)e^{-\hbar \mathcal{H}_{k;r}} e^{-\alpha_1}
        =\sum_{b\in\mathbb{Z}+\frac12}\psi^*_b\,\mathcal{Q}_b(z),
    \end{align}
    where the function $\mathcal{P}_a(w)$ is given by
    \begin{align}\label{eqn:def P}
    \begin{split}
        \mathcal{P}_a(w)&=w^{-a-\frac12}\!\!\sum_{i,t\geq0}
        \frac{\hbar^{i}\,w^{-ik-t}}{i!\,t!\,(r'!)^{i}}
        (\hat{a}+t+ik)_i^{r'},
    \end{split}
    \end{align}
    and the function $\mathcal{Q}_b(z)$ has been defined in equation \eqref{eqn:def Q}.
\end{lemma}
\begin{proof}
Recall that the operator $\mathcal{H}_{k;r}$ is given by
\begin{align*}
    \mathcal{H}_{k;r}
    =\sum_{l\in \mathbb{Z}+\frac{1}{2}}\frac{\hat{l}^{r'}}{r'!}E_{l-k,l}.
\end{align*}
Using equation \eqref{eq: comm E psi} and the Baker--Campbell--Hausdorff formula, we have
\begin{align}
    e^{\hbar\mathcal{H}_{k;r}}\psi_l e^{-\hbar\mathcal{H}_{k;r}}
    =&\sum_{i\geq0}\frac{\hbar^{i}}{i!}\frac{\big(\prod_{j=0}^{i-1}(\hat l-jk)\big)^{r'}}{(r'!)^i}\,\psi_{l-ik},\\
    e^{\hbar\mathcal{H}_{k;r}}\psi^*_m e^{-\hbar\mathcal{H}_{k;r}}
    =&\sum_{p\geq0}\frac{(-\hbar)^{p}}{p!}\frac{\big(\prod_{j=1}^p(\hat m+jk)\big)^{r'}}{(r'!)^p}\,\psi^*_{m+pk}.
\end{align}
Similarly,
recall that $\alpha_1=\sum_{l\in \mathbb{Z}+\frac{1}{2}} E_{l-1,l}$,
we have
\begin{align*}
e^{\alpha_1}\psi_l e^{-\alpha_1}=\sum_{t\geq0}\frac{1}{t!}\psi_{l-t},
\text{\ \ \ and\ \ \ }e^{\alpha_1}\psi^*_m e^{-\alpha_1}=\sum_{t\geq0}\frac{(-1)^t}{t!}\psi^*_{m+t}.
\end{align*}
Composing the above two conjugations and collecting the coefficient of $\psi_a, \psi^*_b$ gives the functions $\mathcal{P}_a(w)$ and $\mathcal{Q}_b(z)$.
This proves this lemma.
\end{proof}

As a consequence,
we can use the formulas \eqref{eqn:conn n=1}, \eqref{eqn:conn n} for connected $N$-point functions of KP tau-functions to compute the connected leaky Hurwitz numbers.

\begin{theorem}\label{thm: conne N point func}
For a given partition $\mu=(\mu_1,\cdots,\mu_N)$, if $l(\mu)=N=1$, we have
\begin{align}\label{eqn:leaky H N1}
    h^{k,r,\circ}_{(n)}(\hbar)
    =\frac{1}{n}\sum_{\substack{a\in\mathbb{N}+\half, \ i,p,t,u\in\mathbb{N}\\ik+t\geq a+\frac{1}{2},\ (i+p)k+t+u=n}}\frac{(-1)^{p+u}\hbar^{i+p}}{i!p!t!u!(r'!)^{i+p}}(-a-\frac{k}{2}+t+ik)_i^{r'}(-a-\frac{k}{2}-u)_p^{r'}.
\end{align}
For $l(\mu)=N>1$, we have
\begin{align}\label{eqn:leaky H N}
    h^{k,r,\circ}_{\mu}(\hbar)
    =\frac{(-1)^{N-1}}{\mathfrak{z}(\mu)}[z_1^{-\mu_1-1}\cdots z_N^{-\mu_N-1}]\sum_{\text{$N$-cycles } \sigma}
	\prod_{i=1}^N \widehat{A}^{leaky}_{k;r} (z_{\sigma^{i}(1)},z_{\sigma^{i+1}(1)}),
\end{align}
where $\mathfrak{z}(\mu)$ is defined in equation \eqref{eq: def zmu} and
\begin{align*}
    \widehat{A}_{k;r}^{leaky}(z_i,z_j)=
    \begin{cases}
        \sum_{a\in\mathbb{N}+\half}\mathcal{P}_{-a}(z_j)\,\mathcal{Q}_{-a}(z_i), &\text{\ if\ }i<j,\\
        -\sum_{a\in\mathbb{N}+\half}\mathcal{P}_{a}(z_j)\,\mathcal{Q}_{a}(z_i), &\text{\ if\ }i>j.
    \end{cases}
\end{align*}
The functions $\mathcal{P}_{a}(z)$ and $\mathcal{Q}_{a}(w)$ are defined in equations \eqref{eqn:def P} and \eqref{eqn:def Q}.

\end{theorem}

\noindent \emph{Proof}.
For the case of $N=1$,
by applying \eqref{eqn:conn n=1} and \eqref{eqn: Aleaky as tildePQ},
we have
\begin{align*}
    \sum_{n\geq1}n h^{k,r,\circ}_{(n)}&(\hbar)z^{-n-1}
    =\sum_{a\in\mathbb{N}+\half}\tilde{\mathcal{P}}_{-a}(z)\,\mathcal{Q}_{-a}(z)\\
    =&\sum_{\substack{a\in\mathbb{N}+\half, \ i,p,t,u\in\mathbb{N}\\ik+t\geq a+\frac{1}{2}}}(-1)^{p+u}\hbar^{i+p}\frac{z^{-(i+p)k-(t+u)-1}}{i!p!t!u!(r'!)^{i+p}}(-a-\frac{k}{2}+t+ik)_i^{r'}(-a-\frac{k}{2}-u)_p^{r'},
\end{align*}
which proves equation \eqref{eqn:leaky H N1}.

For $N\geq2$,
note that when $i<j$,
we have
$$\widehat{A}_{k;r}^{leaky}(z_i,z_j)
=\sum_{a\in\mathbb{N}+\half}\mathcal{P}_{-a}(z_j)\,\mathcal{Q}_{-a}(z_i)$$
where the equality comes from \eqref{eqn: Aleaky as PQ}.
When $i>j$,
we have
\begin{align*}
    \widehat{A}_{k;r}^{leaky}(z_i,z_j)&=\sum_{a\in\mathbb{N}+\half}\mathcal{P}_{-a}(z_j)\,\mathcal{Q}_{-a}(z_i)
    -\sum_{h\geq0}z_i^{-1-h}z_j^{h}-\sum_{h\geq 0} z_j^{-1-h}z_i^h\\
    &=-\langle0| e^{\alpha_1} e^{\hbar \mathcal{H}_{k;r}}  \psi^*(z_i)\psi(z_j) |0\rangle
    =-\sum_{a\in\mathbb{N}+\half}\mathcal{P}_{a}(z_j)\,\mathcal{Q}_{a}(z_i),
\end{align*}
where the last equality follows from Lemma \ref{lem:conj on psipsi*},
and we just need to note that $\langle0|\psi^*_a\psi_b|0\rangle=\delta_{a,b}\delta_{b>0}$.
As a consequence,
by directly applying the formula \eqref{eqn:conn n},
we have
\begin{align*}
    &\sum_{\mu_1,\cdots,\mu_N\geq1}^\infty \mathfrak{z}(\mu)h^{k,r,\circ}_{\mu}(\hbar)z_1^{-\mu_1-1}\cdots z_N^{-\mu_N-1}\\
    =&(-1)^{N-1}\cdot \sum_{\text{$N$-cycles } \sigma}
	\Big( \prod_{i=1}^N \widehat{A}^{leaky}_{k;r} (z_{\sigma^{i}(1)},z_{\sigma^{i+1}(1)}) \Big)
	-\delta_{N,2}\sum_{h\geq0}(h+1)z_1^{-2-h}z_2^h.
\end{align*}
Then equation \eqref{eqn:leaky H N} is obtained by taking the coefficient of $z_1^{-\mu_1-1}\cdots z_N^{-\mu_N-1}$ in the above formula.
Note that the term $\delta_{N,2}\sum_{h\geq0}(h+1)z_1^{-2-h}z_2^h$ does not contribute.\hfill$\square$

\subsection{The genus dependence of leaky Hurwitz numbers}
The $N$-point function derived in the last subsection is convenient to study the genus dependence of leaky Hurwitz numbers (see \cite{Y25}).
In this subsection,
we use it to study the maximal genus of leaky Hurwitz numbers with a given partition $\mu$.
To keep notation concise,
we only consider the case of $k=r=1$ in this subsection.

\begin{lemma}\label{lem:kr1 mach c}
For $k=r=1$,
the coefficient $c^{leaky}_{(m|n)}$ of hook type in the Schur polynomial expansion \eqref{eqn:leaky tau as schur} is a polynomial in $\hbar$ of degree $n+m$.
Moreover,
the top degree term is given by
\begin{align}
    [\hbar^{n+m}] c_{(m|n)}
    =\frac{n! m!}{2^{n+m}}.
\end{align}
\end{lemma}
\begin{proof}
By Proposition \ref{prop:Aleaky formula} and Lemma \ref{lem:conj on psipsi*} with $k=r=1$,
we have
\begin{align*}
    (-1)^n &c_{(m|n)}
    =\langle0|e^{\alpha_1} e^{\hbar \mathcal{H}_{k;r}}  \psi_{m+\half}\psi^*_{-n-\half}|0\rangle\\
    =&\sum_{i,p,t_1,t_2\geq0}
    \frac{\hbar^i\prod_{j=0}^{i-1}(m-j)^2}{i!t_1!2^i}
    \frac{(-\hbar)^p(-1)^{t_2}\prod_{j=0}^{p-1}(n-j)^2}{p!t_2!2^p}\langle0|\psi_{m+\half-i-t_1}\psi^*_{-n-\half+p+t_2}|0\rangle.
\end{align*}
We now analyse the non-zero terms in the above sum.
First,
we have $i\leq m, p\leq n$.
Second,
the formula $\langle0|\psi_{-b-\half}\psi^*_{-a-\half}|0\rangle=\delta_{a,b}\delta_{a\geq0}$
forces $m+\half-i-t_1=-n-\half+p+t_2<0$ in the sum above.
That gives $i+t_1+p+t_2=m+n+1, t_1\geq1$.
As a consequence,
the power of $\hbar$ is $i+p\leq n+m$,
with equality only when $i=m$ and $p=n$,
which implies $t_1=1, t_2=0$.
The corresponding coefficient is
\begin{align*}
    (-1)^n [\hbar^{n+m}] c_{(m|n)}
    =\frac{(m!)^2}{m!2^m}
    \cdot\frac{(-1)^{n}(n!)^2}{n!2^n}
    =(-1)^n\frac{n! m!}{2^{n+m}}.
\end{align*}
\end{proof}

\begin{theorem}\label{thm:genus dependence}
For $k=r=1$ and $\mu=(\mu_1,\dots,\mu_N)$, the connected series $h^{1,1,\circ}_\mu(\hbar)$ is non-zero only if $|\mu|\equiv N\pmod 2$,
then it is a polynomial in $\hbar$ of degree at most $|\mu|-1$.
We denote $g_{\max}:=\frac{|\mu|-N}{2}$,
then $h^{1,1,\circ}_{g,\mu}=0$ if $g>g_{\max}$
and
\begin{align}\label{eqn:top N}
    h^{1,1,\circ}_{g_{\max},\mu}
    =\frac{(|\mu|-1)!}{\mathfrak{z}(\mu)}
    \cdot \frac{\sum_{S\subseteq[N]}(-1)^{|S|+\sum_{j\in S}\mu_j} \big(\sum_{j\in S}\mu_j\big)! \big(\sum_{j\in S^c}\mu_j\big)!}{2^{|\mu|-1}(|\mu|+1)}.
\end{align}
In particular, for $l(\mu)=N=1$,
\begin{align*}
    h^{1,1,\circ}_{g_{\max},(2n+1)}
    =h^{1,1,\circ}_{n,(2n+1)}=\frac{((2n)!)^2}{4^{n}(n+1)},
\end{align*}
and for $l(\mu)=N=2, 3$,
we have
\begin{align*}
    h^{1,1,\circ}_{g_{\max},(\mu_1,\mu_2)}
    =h^{1,1,\circ}_{(|\mu|-2)/2,(\mu_1,\mu_2)}
    =&\frac{(|\mu|-1)!}{\mathfrak{z}(\mu)}
    \cdot \dfrac{|\mu|!-(-1)^{\mu_1}\mu_1! \mu_2!}{(|\mu|+1) 2^{|\mu|-2}},\\
    h^{1,1,\circ}_{g_{\max},(\mu_1,\mu_2,\mu_3)}
    =h^{1,1,\circ}_{(|\mu|-3)/2,(\mu_1,\mu_2,\mu_3)}
    =&\frac{(|\mu|-1)!}{\mathfrak{z}(\mu)}
    \cdot \frac{|\mu|!-\sum_{i=1}^3(-1)^{\mu_i}\mu_i!(|\mu|-\mu_i)!}
    {(|\mu|+1)2^{|\mu|-2}}.
\end{align*}
\end{theorem}

\noindent \emph{Proof.}
First note that in the case of $k=r=1$,
we have $s=g-1+\frac{N+|\mu|}{2}$.
Thus Hurwitz numbers are non-zero only if $N+|\mu|\equiv 0 (\mod 2)$.
When $l(\mu)=N=1$ and $\mu=(\mu_1)$,
the integer $\mu_1$ should be odd.
We directly apply the formula \eqref{eqn:def Azw} and Lemma \ref{lem:kr1 mach c} to obtain that
$h^{1;1;\circ}_{(\mu_1)}(\hbar)$ is a polynomial in $\hbar$ with top term $\hbar^{\mu_1-1}$,
which means the maximal $s_{max}=\mu_1-1$.
Equivalently,
the maximal genus should be
$g_{max}=s_{max}+1-\frac{1+\mu_1}{2}
=\frac{\mu_1-1}{2}$.
We denote $\mu_1=2n+1$,
then
\begin{align*}
    h^{1,1}_{g_{max},(2n+1)}
    =&h^{1,1}_{n,(2n+1)}
    =[\hbar^{2n}] \frac{(2n)!}{2n+1} \sum_{m=0}^{2n} (-1)^{2n-m} c_{(m|2n-m)}\\
    =&\frac{(2n)!}{(2n+1)2^{2n}} \sum_{m=0}^{2n} (-1)^{m} (2n-m)!m!
    =\frac{((2n)!)^2}{4^n(n+1)},
\end{align*}
where we have used the identity $\sum_{m=0}^{2n} (-1)^{m}(2n-m)!m! = \frac{(2n+1)!}{n+1}$. This identity is proved using the Beta function
\begin{equation}\label{eq:betafunction}
    B(m,n)=\int_{0}^1x^{m-1}(1-x)^{n-1}dx=\frac{(m-1)!(n-1)!}{(m+n-1)!}.
\end{equation}

Next,
we deal with the case of $l(\mu)=N\geq2$.
We need to use the formula \eqref{eqn:leaky H N}.
Then each term contributing to the genus-generating function of leaky Hurwitz numbers $h^{1;1;\circ}_{\mu}(\hbar)$ is obtained by choosing an $N$-cycle and computing products of
\begin{align*}
    \widehat{A}_{1,1}^{leaky}(z_{\sigma^i(1)},z_{\sigma^{i+1}(1)})=&A_{1,1}^{leaky}(z_{\sigma^i(1)},z_{\sigma^{i+1}(1)})+\begin{cases}
        \sum\limits_{h\geq 0} z_{\sigma^{i}(1)}^{-1-h}z_{\sigma^{i+1}(1)}^h, &\text{if}\ \sigma^{i}(1)<\sigma^{i+1}(1),\\
        -\sum\limits_{h\geq 0} z_{\sigma^{i+1}(1)}^{-1-h}z_{\sigma^{i}(1)}^h, &\text{if}\ \sigma^{i}(1)>\sigma^{i+1}(1),
    \end{cases}
\end{align*}
for $i=1,\dots,N$.
For each given $N$-cycle $\sigma$ and a subset $I\subseteq[N]$,
define the corresponding contribution by
\begin{align}\label{eqn:Cont_sigmaI}
\begin{split}
    \text{Cont}_{\sigma,I}&
    :=\frac{(-1)^{N-1}s!}{\mathfrak{z}(\mu)}
	\cdot [\prod_{i=1}^N z_i^{-\mu_i-1}]
	\prod_{i\in I} \bigg(\sum_{n_i,m_i\geq0}(-1)^{n_i}c_{(m_i|n_i)}z_{\sigma^i(1)}^{-n_i-1}z_{\sigma^{i+1}(1)}^{-m_i-1}\bigg)\\
	&\ \ \ \ \cdot\prod_{i\notin I}
	\Big(\delta_{\sigma^{i}(1)<\sigma^{i+1}(1)}
	\cdot\sum\limits_{h_i\geq 0} z_{\sigma^{i}(1)}^{-1-h_i}z_{\sigma^{i+1}(1)}^{h_i}
	-\delta_{\sigma^{i}(1)>\sigma^{i+1}(1)}
	\cdot\sum\limits_{h_i\geq 0} z_{\sigma^{i+1}(1)}^{-1-h_i}z_{\sigma^{i}(1)}^{h_i}\Big).
\end{split}
\end{align}
Then we have
$h^{1,1,\circ}_{\mu}(\hbar)=\sum_{\sigma,I}\text{Cont}_{\sigma,I}$.
Note that each term in the second line of equation \eqref{eqn:Cont_sigmaI} always has $z$-exponent $-1$,
and in the contribution $\text{Cont}_{\sigma,I}$,
we only need the term whose total $z$-exponent is
$$\sum_{i=1}^N(-\mu_i-1)=-(|\mu|+N)=-\sum_{i\in I}(n_i+m_i+2)-(N-|I|).$$
Hence we have $\sum_{i\in I}(n_i+m_i)=|\mu|-|I|$,
which is the maximal power of $\hbar$ of that term by Lemma \ref{lem:kr1 mach c}.
The contribution with $I=\emptyset$ always has $z$-exponent $-N\neq -(|\mu|+N)$,
so it must be zero.
Then the maximal power of $\hbar$ cannot be greater than $|\mu|-1$ and the term whose power is $|\mu|-1$ comes from the contributions $\text{Cont}_{\sigma,I}$ satisfying $|I|=1$.
In this case,
the corresponding maximal genus is $g_{max}=|\mu|-\frac{N+|\mu|}{2}=\frac{|\mu|-N}{2}$.

Below,
we explicitly compute the contributions.
We denote $I=\{i_0\}$.
Then in this case the leaky Hurwitz numbers can be represented by
\begin{align}\label{eqn:h11 as i0}
    h^{1,1,\circ}_{g_{\max},\mu}
    =[\hbar^{|\mu|-1}]\sum_{N-cycles\ \sigma} \sum_{i_0=1}^N \text{Cont}_{\sigma,\{i_0\}}.
\end{align}
For each pair $(\sigma,i_0)$,
we associate a permutation $w\in S_N$ such that $w(i)=\sigma^{i+i_0}(1)$.
Actually, $w$ is the unique permutation satisfying $w(N)=\sigma^{i_0}(1)$ and 
$$w(1)\rightarrow w(2) \rightarrow \dots \rightarrow w(N)\rightarrow w(1)
$$
exactly represents the $N$-cycle $\sigma$.
Moreover,
note that from Lemma \ref{lem:kr1 mach c},
the $\hbar$-leading part of $A_{1,1}^{leaky}(z,w)$ factorizes as $U(z)V(w)$ with
\begin{align}
    U(z)=&\sum_{n\geq0}\frac{(-1)^n n!}{2^n}\hbar^n z^{-n-1}
    =\int_0^\infty \frac{2e^{-2u}du}{z+\hbar u},\label{eq:Uz}\\
    V(w)=&\sum_{m\geq0}\frac{m!}{2^m}\hbar^m w^{-m-1}
    =\int_0^\infty \frac{2e^{-2v}dv}{w-\hbar v}.\label{eq:Vz}
\end{align}
These can be seen by replacing the factorials on the left hand side by
$n!=\int_0^\infty t^ne^{-t}dt$. Note that whenever $\hbar$ appears in the denominator, we always regard it as expanded as a power series in $\hbar$, in other words, the expansion is such that all powers of $\hbar$ are nonnegative.
Thus, from equation \eqref{eqn:h11 as i0} and the correspondence between $(\sigma,i_0)$ and $w\in S_N$, and assuming that $|z_1|>|z_2|>\cdots>|z_N|$,
we have
\begin{align*}
    h^{1,1,\circ}_{g_{\max},\mu}
    =&\frac{(-1)^{N-1}s!}{\mathfrak{z}(\mu)}
	\cdot [\hbar^{|\mu|-1} \prod_{j=1}^N z_j^{-\mu_j-1}]
	\sum_{w\in S_N} U(z_{w(N)})V(z_{w(1)})
    \cdot \prod_{i=1}^{N-1} \frac{1}{z_{w(i)}-z_{w(i+1)}}.\\
    =&\frac{(-1)^{N-1}s!}{\mathfrak{z}(\mu)}
	\cdot [\hbar^{|\mu|-1} \prod_{j=1}^N z_j^{-\mu_j-1}]\\
    &\cdot\int_0^\infty\int_0^\infty\sum_{w\in S_N}\frac{-4e^{-2(u+v)}}{(\hbar v-z_{w(1)})
    \cdot \prod_{i=1}^{N-1}(z_{w(i)}-z_{w(i+1)})\cdot (z_{w(N)}+\hbar u)}dudv.
\end{align*}
By Lemma \ref{lem:sum w}, this equals
\begin{align*}
    &\frac{s!}{\mathfrak{z}(\mu)}
	\cdot [\hbar^{|\mu|-1} \prod_{j=1}^N z_j^{-\mu_j-1}]
    \int_{0}^\infty \int_{0}^{\infty} \frac{\hbar^{N-1}(u+v)^{N-1}}{\prod_{j=1}^N (z_j-\hbar v)(z_j+\hbar u)}
    4e^{-2(u+v)} dudv\\
    =&\frac{s!}{\mathfrak{z}(\mu)}
	\cdot [\hbar^{|\mu|-1} ]\int_{0}^\infty \int_{0}^{\infty}\hbar^{N-1}(u+v)^{N-1}4e^{-2(u+v)}\prod_{j=1}^N\frac{(\hbar v)^{\mu_j}-(-\hbar u)^{\mu_j}}{\hbar(u+v)}dudv\\
    =&\frac{s!}{\mathfrak{z}(\mu)}
    \int_{0}^\infty \int_{0}^{\infty} \frac{\prod_{j=1}^N(v^{\mu_j}-(-u)^{\mu_j})}{u+v}
    4e^{-2(u+v)} dudv.
\end{align*}
We expand the term
\begin{equation}\label{eq:expansion}
    \prod_{j=1}^N(v^{\mu_j}-(-u)^{\mu_j})
=\sum_{S\subseteq[N]} v^{\sum_{j\in S^c} \mu_j} (-1)^{\sum_{j\in S}(\mu_j+1)} u^{\sum_{j\in S} \mu_j}
\end{equation}
and use $\frac{1}{u+v}=\int_0^\infty e^{-(u+v)t}dt$ to obtain
\begin{align*}
    h^{1,1,\circ}_{g_{\max},\mu}
    =&\frac{(|\mu|-1)!}{\mathfrak{z}(\mu)}
    \cdot \int_{0}^\infty \int_{0}^{\infty} \int_{0}^{\infty}\sum\limits_{S\subseteq[N]} v^{\sum\limits_{j\in S^c} \mu_j} (-1)^{\sum\limits_{j\in S}(\mu_j+1)} u^{\sum\limits_{j\in S} \mu_j}
    \cdot4e^{-2(u+v)} e^{-(u+v)t}dtdudv\\
    =&\frac{(|\mu|-1)!}{\mathfrak{z}(\mu)}
    \cdot \sum_{S\subseteq[N]} (-1)^{\sum_{j\in S}(\mu_j+1)}\big(\sum_{j\in S}\mu_j\big)! \big(\sum_{j\in S^c}\mu_j\big)!
    \int_0^\infty 4(2+t)^{-|\mu|-2}dt\\
    =&\frac{(|\mu|-1)!}{\mathfrak{z}(\mu)}
    \cdot \frac{\sum_{S\subseteq[N]} (-1)^{\sum_{j\in S}(\mu_j+1)}\big(\sum_{j\in S}\mu_j\big)! \big(\sum_{j\in S^c}\mu_j\big)!}
    {2^{|\mu|-1}(|\mu|+1)},
\end{align*}
where we have used $\int_0^\infty u^a e^{-bu}du=a!b^{-a-1}$.\hfill$\square$

\begin{lemma}\label{lem:sum w}
For any positive integer $N$,
we have
\begin{align}\label{eqn: sum w}
    \sum_{w\in S_N}\frac{(a-b)}{(a-z_{w(1)})\cdot \prod\limits_{i=1}^{N-1}(z_{w(i)}-z_{w(i+1)})\cdot (z_{w(N)}-b)}
    =\prod_{j=1}^N\frac{(a-b)}{(a-z_j)(z_j-b)}.
\end{align}
\end{lemma}
\begin{proof}
We prove this lemma by induction on $N$.
The $N=1$ case is trivial.
We assume this lemma holds for $N-1$.
As a function of $a$,
the left hand side of equation \eqref{eqn: sum w}
\begin{align*}
    L_N(a,b;z_1,\dots,z_N)
    :=\sum_{w\in S_N}\frac{(a-b)}{(a-z_{w(1)})\cdot \prod_{i=1}^{N-1}(z_{w(i)}-z_{w(i+1)})\cdot (z_{w(N)}-b)}
\end{align*}
only has possible simple poles at $z_k, k=1,\dots, N$.
Note that when computing its residue at the point $a=z_k$,
only the permutations $w\in S_N$ satisfying $w(1)=k$ make non-zero contribution.
Thus,
\begin{align*}
    \text{Res}_{a=z_k} L_N(a,b;&z_1,\dots,z_N)
    =\sum_{w\in S_N \atop w(1)=k}\frac{(z_k-b)}{\prod_{i=1}^{N-1}(z_{w(i)}-z_{w(i+1)})\cdot (z_{w(N)}-b)}\\
    =&L_{N-1}(z_k,b;z_1,\dots,\widehat{z_k},\dots,z_N)
    =\prod_{j=1, j\neq k}^{N}\frac{(z_k-b)}{(z_k-z_j)(z_j-b)},
\end{align*}
where the last equality follows from the induction hypothesis.
For the right hand side of equation \eqref{eqn: sum w},
by direct computation,
we also have
\begin{align*}
    \text{Res}_{a=z_k} \prod_{j=1}^N\frac{(a-b)}{(a-z_j)(z_j-b)}
    =\prod_{j=1, j\neq k}^{N}\frac{(z_k-b)}{(z_k-z_j)(z_j-b)}.
\end{align*}
Therefore the difference between two sides of equation \eqref{eqn: sum w} has no finite poles in $a$,
hence is a polynomial $P(a)$ in $a$.
As $a\rightarrow\infty$, both sides tend to finite limits, so $P(a)$ is constant.
Setting $a=b$, both sides are equal to zero,
thus we have $P(a)\equiv0$.
This lemma is then proved.
\end{proof}

The following proposition completes the proof of Theorem \ref{thm:main thm 3}.

\begin{proposition}
The leaky Hurwitz number $h^{1,1,\circ}_{g_{\max},\mu}$ with $g_{\max}=\frac{|\mu|-N}{2}$ is
strictly positive.
In particular, the genus-generating series $h^{1,1,\circ}_{\mu}(\hbar)$ has degree exactly $|\mu|-1$ in $\hbar$.
\end{proposition}

\begin{proof}
We will use equation \eqref{eqn:top N} to compute $h^{1,1,\circ}_{g_{\max},\mu}$.
Put
\begin{align*}
\Sigma(\mu):=
\sum_{S\subseteq[N]} (-1)^{|S|+\sum_{j\in S}\mu_j}
(\sum_{j\in S}\mu_j)! (\sum_{j\in S^c}\mu_j)!.
\end{align*}
Since other factors in \eqref{eqn:top N} is positive, 
we need only show that $\Sigma(\mu)>0$. 
Using equation \eqref{eq:betafunction}, we obtain
\begin{align*}
\Sigma(\mu)&=(d+1)!\int_0^1
\sum_{S\subseteq[N]}(-1)^{|S|+\sum_{j\in S}\mu_j}x^{\sum_{j\in S}\mu_j}
(1-x)^{\sum_{j\in S^c}\mu_j}\,dx \\
&=(d+1)!\int_0^1\prod_{j=1}^N\big((1-x)^{\mu_j}-(-x)^{\mu_j}\big)\,dx,
\end{align*}
where the second equality is obtained by letting $v=1-x$ and $u=x$ in equation \eqref{eq:expansion}.

Since $|\mu|-N=\sum_{j=1}^N(\mu_j-1)$ is even, there are even number of even $\mu_j$. If $\mu_j$ is odd, then
\begin{equation}\label{eq:component}
(1-x)^{\mu_j} + (-1)^{\mu_j+1}x^{\mu_j}
\end{equation}
is positive for every $x\in(0,1)$. If $\mu_j$ is even, then \eqref{eq:component} is positive on
$(0,\frac{1}{2})$, and is negative on
$(\frac{1}{2},1)$. So the product
\[\prod_{j=1}^N
\big((1-x)^{\mu_j} - (-x)^{\mu_j} \big)\]
is positive away from $1/2$. Therefore, we have $\Sigma(\mu)>0$.
\end{proof}

\bibliographystyle{plain}
\bibliography{reference}

\vspace{30pt} \noindent
Chongyu Wang \\
Academy of Mathematics and Systems Science, \\
Chinese Academy of Sciences, Beijing, China. \\
Email: {\it wangcyu@pku.edu.cn}

\vspace{30pt} 
\noindent Chenglang Yang \\
Institute for Math and AI, Wuhan University, Wuhan, China.\\
Email: {\it yangcl@whu.edu.cn}

\end{document}